\documentclass[review,1p,times,twocolumn]{elsarticle}

\usepackage{amsfonts}
\usepackage{amsmath}
\usepackage{amssymb}
\usepackage{amsthm}
\usepackage[ruled,titlenumbered,lined]{algorithm2e}
\usepackage{varioref}
\usepackage{pdflscape}
\usepackage{lscape}
\usepackage{color}
\usepackage{subfig}

\usepackage{caption}

\usepackage{booktabs}
\usepackage{multirow}
\usepackage{multicol}

\newtheorem{definition}{Definition}
\newtheorem{property}{Property}
\newtheorem{theorem}{Theorem}

\newtheorem{proposition}{Proposition}

\makeatletter
\newcommand*{\rom}[1]{\expandafter\@slowromancap\romannumeral #1@}
\makeatother

\journal{arXiv.org}

\begin{document}

\begin{frontmatter}



\title{A linear programming based heuristic framework for min-max regret combinatorial optimization problems with interval costs}


\author[dep]{Lucas Assun\c{c}\~ao\corref{cor1}}
\ead{lucas-assuncao@ufmg.br}
\author[dcc]{Thiago F. Noronha}
\ead{tfn@dcc.ufmg.br}
\author[utt]{Andr\'ea Cynthia Santos}
\ead{andrea.duhamel@utt.fr}
\author[ufc]{Rafael Andrade}
\ead{rca@lia.ufc.br}
\cortext[cor1]{Corresponding author.}

\address[dep]{
  Departamento de Engenharia de Produ\c{c}\~ao, Universidade Federal de Minas Gerais\\
  Avenida Ant\^onio Carlos 6627, CEP 31270-901, Belo Horizonte, MG, Brazil\\
}

\address[dcc]{
  Departamento de Ci\^encia da Computa\c{c}\~ao, Universidade Federal de Minas Gerais\\
  Avenida Ant\^onio Carlos 6627, CEP 31270-901, Belo Horizonte, MG, Brazil\\
}

\address[utt]{
  ICD-LOSI, UMR CNRS 6281, Universit\'e de Technologie de Troyes\\
  12, rue Marie Curie, CS 42060, 10004, Troyes CEDEX, France\\
}

\address[ufc]{
  Departamento de Estat\'{\i}stica e Matem\'{a}tica Aplicada, Universidade Federal do Cear\'{a}\\
  Campus do Pici - Bloco 910, 60455-900, Fortaleza, CE, Brazil\\
}

\begin{abstract}
  This work deals with a class of problems under interval data
  uncertainty, namely \emph{interval robust-hard} problems, composed of
  interval data min-max regret generalizations of classical
  NP-hard combinatorial problems modeled as 0-1 integer linear programming problems. These problems are more challenging than other interval data
  min-max regret problems, as solely computing the cost of any feasible
  solution requires solving an instance of an NP-hard problem.
  The state-of-the-art exact algorithms in the literature are based on the generation
  of a possibly exponential number of cuts. As each cut
  separation involves the resolution of an NP-hard classical optimization problem,
  the size of the instances that can be solved efficiently is
  relatively small.
  To smooth this issue, we present a modeling technique for interval robust-hard problems in the context of a heuristic framework.
  The heuristic obtains feasible solutions by exploring dual information of a linearly relaxed
  model associated with the classical optimization problem counterpart.
  Computational experiments for interval data min-max regret versions of the restricted shortest path problem and the set covering problem show
  that our heuristic is able to find optimal or near-optimal solutions
  and also improves the primal bounds obtained by a state-of-the-art
  exact algorithm and a 2-approximation procedure for interval data min-max regret problems.
\end{abstract}


\end{frontmatter}
\section{Introduction}
\label{s_intro}
Robust optimization \cite{Kouvelis1997} is an alternative to stochastic
programming \cite{Spall03} in which the variability of the data is represented 
by deterministic values in the context of \emph{scenarios}.
A scenario corresponds to a parameters assignment, i.e., a value is fixed for each parameter subject to uncertainty.
The two main approaches adopted to model robust optimization problems are the \emph{discrete scenarios model} and the \emph{interval data model}.
In the former, a discrete set of possible scenarios is considered. In the latter, the uncertainty referred to a parameter is represented by a
continuous interval of possible values. Differently from the discrete scenarios model, the infinite many possible scenarios that arise
in the interval data model are not explicitly given.
Nevertheless, in both models, a classical (i.e., parameters known in advance) optimization problem takes place
whenever a scenario is established.

With respect to robust optimization criteria, the \textit{min-max regret} (also known as \emph{robust deviation}
criterion) is one of the most used in the literature (see, e.g.,
\cite{Averbakh05,Mo06,MontemanniBarta07,Pereira11}). The
\emph{regret} of a solution in a given scenario is defined as the
cost difference between such solution and an optimal one in this scenario. In turn, the robustness
cost of a solution is defined as its maximum regret over all
scenarios. Min-max regret problems aim at finding a solution with the
minimum robustness cost, which is referred to as a \textit{robust solution}.
We refer to \cite{Kouvelis1997} for details on other robust optimization criteria.

Robust optimization versions of several combinatorial optimization
problems have been studied in the literature, addressing, for example, uncertainties over costs. Handling uncertain
costs brings an extra level of difficulty, such that even polynomially solvable problems become NP-hard in their corresponding robust
versions \cite{Mo06,Pereira11,YaKaPi01,MonteGa04c, Kasperski08}. A recent trend in this field is to investigate robust optimization
problems whose classical counterparts are already
NP-hard. We refer to these problems as \textit{robust-hard} problems.
The interval data min-max regret restricted shortest path problem, introduced in this study, belongs to this class of problems, along
with interval data min-max regret versions of the traveling salesman problem \cite{MontemanniBarta07}, the 0-1 knapsack problem \cite{Deineko2010} and
the set covering problem \cite{Averbakh13}.

This work addresses \emph{interval min-max regret} problems, which consist of min-max regret versions of combinatorial problems with interval costs.
Precisely, we consider interval data min-max regret generalizations of classical combinatorial problems modeled by means of 0-1 Integer Linear Programming
(ILP). Notice that a large variety of combinatorial problems are modeled as 0-1 ILP problems, including
(i) polynomially solvable problems, such as the shortest path problem, the minimum spanning tree
problem and the assignment problem, and (ii) NP-hard problems, such as the 0-1 knapsack problem, the set
covering problem, the traveling salesman problem and the restricted shortest path problem \cite{Garey79,Handler80}.
In this study, we are particularly interested in a subclass of interval min-max regret problems, namely
\emph{interval robust-hard problems}, composed of interval data min-max regret versions of classical NP-hard combinatorial problems
as those aforementioned in (ii).

Aissi at al.~\cite{Aissi09} (also see \cite{Averbakh01}) showed that, for any interval min-max regret problem (including interval robust-hard problems),
the robustness cost of a solution can be computed by solving a single instance of the classical optimization problem counterpart
in a particular scenario. Therefore, one does not have to consider all the infinite many possible scenarios during the search for a robust solution,
but only a subset of them, one for each feasible solution.
Notice, however, that, in the case of interval robust-hard problems, computing the cost of a solution still involves solving an NP-hard problem.
Thus, these problems are more challenging than other interval min-max regret problems.

The state-of-the-art exact
algorithms used to solve these problems are based on the
generation of a possibly exponential number of cuts (see, e.g.,
\cite{MontemanniBarta07,Averbakh13}). Since each cut separation
implies solving an instance of an NP-hard problem (which corresponds
to the classical counterpart of the robust optimization problem
considered), the size of the instances efficiently solvable is considerably smaller than
that of the instances solved for the corresponding classical counterparts.
In fact, to the best of our knowledge, no work in the literature
provides a compact ILP formulation (with a polynomial number of variables and constraints) for any interval robust-hard problem.

Although we do not close the aforementioned issues in this work, we smooth them by presenting an alternative way to handle interval robust-hard problem.
We propose a Linear Programming (LP) based heuristic framework inspired by the modeling technique introduced in
\cite{YaKaPi01,Karasan01} and formalized by Kasperski \cite{Kasperski08}.
The heuristic framework is suitable to tackle interval min-max regret problems in general and consists in solving a
Mixed Integer Linear Programming (MILP) model based on the dual of a linearly relaxed formulation for the classical optimization problem counterpart.
An optimal solution for the model solved by our heuristic is not necessarily optimal for the original robust optimization
problem. However, computational experiments for interval data min-max regret versions of the restricted shortest path problem and the set covering
problem \cite{Averbakh13} showed that the heuristic is able to find optimal or near-optimal
solutions, and that it improves the upper bounds obtained by both a state-of-the-art exact algorithm and a 2-approximation procedure for interval min-max regret problems.

The theoretical quality of the solutions obtained by the heuristic in general is beyond the scope of this work, as it may rely on specific problem
dependent factors, such as the strength of the formulation adopted to model the classical counterpart problem within the heuristic framework.
In fact, as we will discuss in the sequel, the study of approximative procedures for interval min-max regret problems is a relatively
unexplored field, and not much is known.

The remainder of this work is organized as follows. Related works on robust-hard problems are discussed in Section~\ref{s_rw}.
In Section~\ref{s_general_robust_pb}, we present a standard modeling technique for interval min-max regret problems and a state-of-the-art
logic-based Benders' framework for solving problems in this class.
In this study, the solutions obtained by this exact approach are used to evaluate the quality of the solutions found by the proposed
LP based heuristic framework, which is presented in Section~\ref{s_heuristic}.
In Section~\ref{s_case_studies}, we formally describe the interval robust-hard problems used as case studies for the heuristic framework.
In this same section, we introduce the interval data min-max regret restricted shortest path problem, along with a brief motivation and related works.
Computational experiments are showed in Section~\ref{s_results}, while concluding remarks and future work directions are discussed in the last section.
\section{Related works}
\label{s_rw}
As far as we know, Montemanni et al. \cite{MontemanniBarta07} (also see \cite{MonteTech06}) were the first to address an
interval robust-hard problem. The authors introduced the interval data min-max regret traveling salesman problem, along with a
mathematical formulation containing an exponential number of constraints. Moreover, three exact algorithms were presented and computationally compared
at solving the proposed formulation: a \emph{branch-and-bound}, a \emph{branch-and-cut}
and a logic-based Benders' decomposition~\cite{Hooker03} algorithm. The latter algorithm
has been widely used to solve robust optimization problems whose
classical counterparts are polynomially solvable (see, e.g., \cite{Mo06,Pereira11,MonteGa04c}).
Computational experiments showed that the logic-based Benders' algorithm outperforms the other exact algorithms for the interval data min-max regret
traveling salesman problem.

Later on, Pereira and Averbakh \cite{Averbakh13} introduced the interval data min-max regret set covering problem. The authors proposed a mathematical
formulation for the problem that is similar to the one proposed by \cite{MontemanniBarta07} for the interval data min-max regret traveling salesman problem.
As in \cite{MontemanniBarta07}, the formulation has an exponential number of constraints.
They also adapted the logic-based Benders' algorithm of \cite{Mo06,MontemanniBarta07,MonteGa04c} to this problem and presented
an extension of the method that aims at generating multiple Benders' cuts per iteration of the algorithm. Moreover, the work presents an exact
approach that uses Benders' cuts in the context of a \emph{branch-and-cut} framework. Computational experiments showed that such approach, as well as the
extended logic-based Benders' algorithm, outperforms the standard logic-based Benders' algorithm at solving the instances considered. 
This robust version of the set covering problem was also addressed in \cite{Amadeu15}, where the authors propose scenario-based heuristics with
path-relinking.

A few works also deal with robust optimization versions of the 0-1 knapsack problem.
For instance, the studies \cite{Kouvelis1997} and \cite{Yu96} address a version of the problem where the uncertainty over each item profit is represented by a
discrete set of possible values. In these works, the absolute robustness criterion is considered. In \cite{Yu96}, the author proved that this
version of the problem is strongly NP-hard when the number of possible scenarios is unbounded and pseudo-polynomially solvable for a bounded number of
scenarios. Kouvelis et al.~\cite{Kouvelis1997} also studied a min-max regret version of the problem that considers a discrete set of scenarios of item
profits. They provided a pseudo-polynomial algorithm for solving the problem when the number of possible scenarios is bounded.
When the number os scenarios is unbounded, the problem becomes strongly NP-hard and there is no approximation scheme for it \cite{Aissi07}.

More recently, Feizollahi and Averbakh~\cite{Feizollahi2014} introduced the min-max regret quadratic assignment problem with interval flows,
which is a generalization of the classical quadratic assignment problem in which material flows between facilities are uncertain and vary in given intervals. 
Although quadratic, the problem presents a structure that is very similar to the robust versions of combinatorial problems
we address in this work. The authors proposed two mathematical formulations and adapted the logic-based Benders' algorithm of
\cite{Mo06,MontemanniBarta07,MonteGa04c} to solve them through the linearization of the corresponding master problems.
They also developed a hybrid approach which combines Benders' decomposition with heuristics.

Regarding heuristics for interval robust-hard problems, a simple and efficient
scenario-based procedure to tackle interval min-max regret problems in general was proposed in~\cite{Kasperski05}
and successfully applied in several works (see, e.g., \cite{MonteTech06,Kasperski05,Kasperski06b}).
The procedure, called Algorithm Mean Upper (AMU), consists in solving the corresponding classical optimization problem in
two specific scenarios: the so called \emph{worst-case scenario},
where the cost referred to each binary variable is set to its upper
bound, and the \emph{mid-point scenario}, where the cost of the binary variables
are set to the mean values referred to the bounds of the respective cost intervals.
With this heuristic, one can obtain a feasible solution
for any interval min-max regret optimization problem (including interval robust-hard
problems) with the same worst-case asymptotic complexity of solving
an instance of the classical optimization problem
counterpart. Moreover, it is proven in \cite{Kasperski06b} that this
algorithm is 2-approximative for any interval min-max regret optimization
problem. Notice that AMU does not run in polynomial time for interval robust-hard problems, unless P = NP.

The study of approximative procedures for interval min-max regret problems in general is a relatively
unexplored field, and not much is known. Conde~\cite{Conde2010} proved that AMU gives a 2-approximation for an even broader class of min-max regret
problems. Precisely, while the works of Kasperski et al.~\cite{Kasperski05,Kasperski06b} only address
combinatorial (finite and discrete) min-max regret optimization models, Conde~\cite{Conde2010} extends the 2-approximation result for models with compact
constraint sets in general, including continuous ones.
Recently, some research has been conducted to refine the constant factor of the aforementioned approximation. For example, in~\cite{Conde2012}, the author
attempts to tighten this factor of 2 through the resolution of a robust optimization problem over a reduced uncertainty cost set.
Moreover, in \cite{Chassein2015}, the authors introduced a new bound that gives an instance dependent performance guarantee of the \emph{mid-point scenario}
solution for interval data min-max regret versions of combinatorial problems. They show that the new performance ratio is at most 2, and the bound is successfully applied to solve
the interval data min-max regret shortest path problem~\cite{Kouvelis1997} within a \emph{branch-and-bound} framework.

Kasperski and Zieli\'{n}ski~\cite{Kasperski07} also developed a Fully Polynomial Time Approximation Scheme (FPTAS)
for interval min-max regret problems. However, this FPTAS relies on two very restrictive conditions: (i) the problem tackled must present a
pseudopolynomial algorithm, and (ii) the corresponding classical counterpart has to be polynomially solvable. Condition (ii) comes from the fact that
the aforementioned FPTAS uses AMU within its framework. Notice that, from (ii), this FPTAS does not naturally hold for
interval robust-hard problems.

The existence of efficient approximative algorithms is closely related to another almost unexplored aspect of interval robust-hard
problems, which is their computational complexity. Although NP-hard by definition, whether these problems are necessarily strongly NP-hard
or not is still an open issue. 
\section{Modeling and solving interval min-max regret problems}
\label{s_general_robust_pb}
In this section, we discuss a standard modeling technique in the literature of interval min-max regret problems, as well as a state-of-the-art
exact algorithm to solve problems in this class. To this end, consider $\mathcal{G}$, a generic 0-1 ILP combinatorial problem defined as follows.
\begin{eqnarray}
  \mbox{$(\mathcal{G})\quad$}\min && cy \label{bilp00}\\
  s.t.  && Ay \geq b \label{bilp01}\\
	&& y \in \{0,1\}^n \label{bilp02}.
\end{eqnarray}
The binary variables are represented by an $n$-dimensional column vector $y$, whereas their corresponding cost values are given by an
$n$-dimensional row vector $c$. Moreover, $b$ is an $m$-dimensional
column vector, and $A$ is an $m \times n$ matrix.
The feasible region of $\mathcal{G}$ is given by $\Omega = \{y:\, Ay \geq b, \, y \in \{0,1\}^n\}$.

Let $\mathcal{R}$ be an interval data min-max regret
robust optimization version of $\mathcal{G}$, where a continuous cost interval $[l_i,u_i]$, with $l_i,u_i \in \mathbb{Z}_+$ and $l_i \leq u_i$,
is associated with each binary variable $y_i$, $i =1,\ldots, n$. The following definitions describe $\mathcal{R}$ formally.
\begin{definition}
A \emph{scenario} $s$ is an assignment of costs to the binary variables, i.e., a cost $c_{i}^s \in [{l}_{i},{u}_{i}]$ is fixed for all $\, y_i$, $i =1,\ldots, n$.
\end{definition}
Let $\mathcal{S}$ be the set of all possible cost scenarios, which consists of the cartesian product of the continuous
intervals $[l_i,u_i]$, $i = 1, \ldots, n$.
The cost of a solution $y \in \Omega$ in a scenario $s \in \mathcal{S}$ is given by $c^sy = \sum\limits_{i = 1}^{n}{c^s_iy_i}$.
\begin{definition}
A solution $opt(s) \in \Omega$ is said to be \emph{optimal} for a scenario $s \in \mathcal{S}$ if it has the
smallest cost in $s$ among all the solutions in $\Omega$, i.e., $opt(s) = \arg\min\limits_{y \in \Omega}{c^sy}$.
\end{definition}
\begin{definition}
The \emph{regret (robust deviation)} of a solution $y \in \Omega$ in a scenario $s \in \mathcal{S}$, denoted by $r_y^s$,
is the difference between the cost of $y$ in $s$ and the cost of $opt(s)$ in $s$, i.e.,
$r_y^s = c^sy - c^s{opt(s)}$.
\end{definition}
\begin{definition}
The \emph{robustness cost} of a solution $y \in \Omega$, denoted by $R_y$, is the maximum
regret of $y$ among all possible scenarios, i.e., $R_y = \max\limits_{s \in \mathcal{S}}{r^{s}_y}$.
\end{definition}
\begin{definition}
A solution $y^* \in \Omega$ is said to be \emph{robust} if it has the smallest robustness cost
among all the solutions in $\Omega$, i.e., $y^* = \arg\min\limits_{y \in \Omega}{R_{y}}$.
\end{definition}
\begin{definition}
The \emph{interval data min-max regret problem} $\mathcal{R}$ consists in finding a robust solution.
\end{definition}
For each scenario $s \in \mathcal{S}$, let $\mathcal{G}(s)$ denote the corresponding 0-1 ILP problem $\mathcal{G}$ under the cost vector $c^s \in \mathbb{R}_+^n$
referred to $s$.
Also consider an $n$-dimensional column vector $x$ of binary variables. Then, $\mathcal{R}$ can be generically modeled as follows.
\begin{eqnarray}
  \mbox{$(\mathcal{R})\quad$}\min && \max\limits_{s\in \mathcal{S}}\;(c^sy - \overbrace{\min\limits_{x \in \Omega}{c^sx}}^{\mathcal{G}(s)}) \label{g00} \\
  s.t. && y \in \Omega \label{d00}.
\end{eqnarray}

\begin{theorem}[Aissi et al. \cite{Aissi09}]
\label{teo01}
The robust deviation (regret) of any feasible solution ${\bar y} \in \Omega$ is maximum in the scenario $s({{\bar y}})$ induced by ${\bar y}$, defined as follows:
\[
\textrm{for all } \; i \in \{1,\ldots,n\},\quad c^{s({{\bar y}})}_i = \left\{
\begin{array}{lll}
        \; u_i, & \textrm{if $\,{{\bar y}}_i = 1$,} \\
        \; l_i, & \textrm{if $\,{{\bar y}}_i = 0$.} 
\end{array}\right. \]
\end{theorem}
Theorem~\ref{teo01}, which is also stated in \cite{Averbakh01}, reduces the number of scenarios to be considered during the search for a robust solution.
Accordingly, ${\mathcal{R}}$ can be rewritten taking into account only the scenario induced by the solution that the
$y$ variables define. This scenario, referred to as $s(y)$, is represented by the cost vector $c^{s(y)} ={\Big( l_1 + (u_1 - l_1)y_1, \ldots, l_n + (u_n - l_n)y_n \Big)}$.
Then, $\mathcal{R}$ can be rewritten as 
\begin{eqnarray}
  \mbox{$(\tilde{\mathcal{R}})\quad$}\min && \Big(c^{s(y)}y - \overbrace{\min\limits_{x \in \Omega}{c^{s(y)}x}}^{{\mathcal{G}}({s({y})})}\Big) \label{r00} \\
  s.t. && y \in \Omega \label{r01}.
\end{eqnarray}

In order to obtain an MILP formulation for $\mathcal{R}$, we reformulate $\tilde{\mathcal{R}}$ according to \cite{Aissi09}. Precisely, we add a free
variable $\rho$ and linear constraints that explicitly bound $\rho$ with respect to all the feasible solutions that $x$ can represent.
The resulting MILP formulation is provided from \eqref{milp00} to \eqref{milp03}.
\begin{eqnarray}
  \mbox{$(\mathcal{F})\quad$}\min && (\sum\limits_{i = 1}^{n}{u_iy_i} - \rho) \label{milp00} \\
  s.t. && \rho \leq {\sum\limits_{i = 1}^{n}{(l_i + (u_i - l_i)y_i)\bar{x}_i}} \quad \forall\, \bar{x} \in \Omega, \label{milp01} \\
        && y \in \Omega,   \label{milp02} \\
        && \rho \mbox{ free}. \label{milp03}
\end{eqnarray}

Constraints (\ref{milp01}) ensure that $\rho$ does not exceed the value related to the inner minimization in (\ref{r00}).
Note that, in (\ref{milp01}), $\bar{x}$ is a constant vector, one for each solution in $\Omega$. These constraints are tight whenever
$\bar{x}$ is optimal for the classical counterpart problem $\mathcal{G}$ in the scenario $s(y)$.
Constraints (\ref{milp02}) define the feasible region referred to the $y$ variables, and constraint (\ref{milp03}) gives the domain of the variable $\rho$.

The number of constraints (\ref{milp01}) corresponds to the number of feasible solutions in $\Omega$.
As the size of this region may grow exponentially with the number of binary variables,
this fomulation is particularly suitable to be handled by decomposition methods, such as the logic-based Benders' decomposition \cite{Hooker03} algorithm detailed below.
\subsection{A logic-based Benders' decomposition algorithm}
\label{s_benders}
Benders' decomposition method was originally proposed in \cite{Benders62} (also see \cite{Geoffrion72}) to tackle MILP problems by exploring
duality theory properties. Several methodologies were later studied to improve the convergence speed of the method (see, e.g., \cite{McDaniel77,Magnanti81,Fischetti10}).
More recently, Hooker and Ottosson \cite{Hooker03} introduced the idea of a \emph{logic-based Benders' decomposition}, which is
a Benders-like decomposition approach suitable for a broader class of problems. This specific method is intended to address any
optimization problem by devising valid cuts solely based on logic inference.

The logic-based Benders' algorithm presented below is a state-of-the-art exact method for interval min-max regret problems
\cite{Mo06,MontemanniBarta07,Pereira11,MonteGa04c,Averbakh13}. The algorithm solves formulation $\mathcal{F}$, given by (\ref{milp00})-(\ref{milp03}),
by assuming that, since several of constraints~(\ref{milp01}) might be inactive at optimality, they can be generated on demand whenever they are violated.
The procedure is described in Algorithm~\ref{Benders01}.
Let ${\Omega}^{\psi} \subseteq {\Omega}$ be the set of solutions $\bar{x} \in \Omega$ (Benders' cuts) available at an iteration $\psi$. Also let
$\mathcal{F}^{\psi}$ be a relaxed version of $\mathcal{F}$ in which constraints (\ref{milp01}) are replaced by
\begin{equation}
 \rho \leq \sum\limits_{i = 1}^{n}{( l_{i} + (u_{i} - l_{i})y_{i}){\bar x}_{i}} \quad \forall\, {\bar x} \in {\Omega}^{\psi}. \label{milp04}
\end{equation}

\noindent Thus, the relaxed problem $\mathcal{F}^{\psi}$, called \emph{master problem}, is defined by (\ref{milp00}), (\ref{milp02}), (\ref{milp03}) and (\ref{milp04}).

Let $ub^{\psi}$ keep the best upper bound found (until an iteration $\psi$) on the solution of $\mathcal{F}$. 
Notice that, at the beginning of Algorithm~\ref{Benders01}, ${\Omega}^{1}$ contains the initial Benders' cuts available, whereas $ub^1$ keeps
the initial upper bound on the solution of $\mathcal{F}$.
In this case, ${\Omega}^{1}=\emptyset$ and $ub^1 := + \infty$.
At each iteration ${\psi}$, the algorithm obtains a solution by solving
a corresponding master problem  $\mathcal{F}^{\psi}$ and seeks a constraint (\ref{milp01}) that is most violated by this solution.
Initially, no constraint (\ref{milp04}) is considered, since ${\Omega}^{1}=\emptyset$. An initialization step is then necessary
to add at least one solution to ${\Omega}^{1}$, thus avoiding unbounded solutions during the first resolution of the master problem.
To this end, it is computed an optimal solution for the worst-case scenario $s_u$, in which $c^{s_u} = u$ (Step \rom{1}, Algorithm~\ref{Benders01}).
\begin{algorithm}[!ht]
\caption{Logic-based Benders' algorithm.}
\KwIn{Cost intervals $[l_i,u_i]$ referred to $y_i$, $i =1,\ldots, n$.}
\KwOut{A robust solution for $\mathcal{F}$, and its corresponding robustness cost.}
  $\psi \leftarrow 1$; $ub^1 \leftarrow + \infty$; ${\Omega}^{1} \leftarrow \emptyset$;\\
 \textbf{Step \rom{1}. (Initialization)} \\
 Find an optimal solution ${\bar x}^1 = opt(s_{u})$ for the worst-case scenario $s_{u}$;\\
  ${\Omega}^{1} \leftarrow {\Omega}^{1} \cup \{{\bar x}^1\}$;\\
  \textbf{Step \rom{2}. (Master problem)}\\ Solve the relaxed problem $\mathcal{F}^{\psi}$, obtaining a solution $({\bar y}^{\psi},{\bar \rho}^{\psi})$;\\
  \textbf{Step \rom{3}. (Slave problem)}\\ Find an optimal solution ${\bar x}^{\psi} = opt(s({{\bar y}^{\psi}}))$ for the scenario
  $s({{\bar y}^{\psi}})$ induced by ${\bar y}^{\psi}$
  and use it to compute $R_{{\bar y}^{\psi}}$, the robustness cost of ${\bar y}^{\psi}$;\\
 \textbf{Step \rom{4}. (Stopping condition)} \\
 $lb^{\psi} \leftarrow {u}{\bar y}^{\psi} - {\bar \rho}^{\psi}$;\\
 \If{$lb^{\psi} \geq {R_{{\bar y}^{\psi}}}$}{Return $({\bar y}^{*},R^*)$;}
 \Else{
 $ub^{\psi+1} \leftarrow ub^{\psi} \leftarrow \min \{ub^{\psi}, R_{{\bar y}^{\psi}}\}$;\\
 ${\Omega}^{\psi+1} \leftarrow {\Omega}^{\psi} \cup \{{\bar x}^{\psi}\}$;\\
 $\psi \leftarrow \psi + 1$;\\
 Go to Step \rom{2};}

    \label{Benders01}
\end{algorithm}
\decmargin{1em}

After the initialization step, at each iteration ${\psi}$, the corresponding relaxed problem $\mathcal{F}^{\psi}$
is solved (Step \rom{2}, Algorithm \ref{Benders01}), obtaining a solution $({\bar y}^{\psi},{\bar \rho}^{\psi})$.
Then, the algorithm checks if $({\bar y}^{\psi},{\bar \rho}^{\psi})$ violates any constraint (\ref{milp01}) of the original problem $\mathcal{F}$.
For this purpose, it is solved a \emph{slave problem} that computes $R_{{\bar y}^{\psi}}$ (the actual robustness cost of ${\bar y}^{\psi}$) by finding an optimal solution
${\bar x}^{\psi} = opt(s({{\bar y}^{\psi})})$ for the scenario $s({{\bar y}^{\psi})}$ induced by ${\bar y}^{\psi}$
(see Step \rom{3}, Algorithm \ref{Benders01}). Notice that each slave problem involves solving the classical optimization problem $\mathcal{G}$,
given by (\ref{bilp00})-(\ref{bilp02}), in the scenario $s({{\bar y}^{\psi}})$.

Let $lb^{\psi} = {u}{{\bar y}^{\psi}} - {\bar \rho}^{\psi}$ be the value of the objective function in (\ref{milp00}) related
to the solution $({\bar y}^{\psi},{\bar \rho}^{\psi})$ of the current master problem $\mathcal{F}^{\psi}$.
According to \cite{Assuncao16b}, $lb^{\psi}$ and $R_{{\bar y}^{\psi}}$ give, respectively, a lower (dual) and an upper (primal) bounds on the solution of $\mathcal{F}$.
If $lb^{\psi}$ does not reach $R_{{\bar y}^{\psi}}$, $ub^{\psi}$ and $ub^{\psi+1}$ are both set to the best upper bound found by the algorithm until
the iteration $\psi$. In addition, a new constraint (\ref{milp04}) is generated from ${\bar x}^{\psi}$ and added to $\mathcal{F}^{\psi+1}$ by setting
${\Omega}^{\psi+1} \leftarrow {\Omega}^{\psi} \cup \{{\bar x}^{\psi}\}$. Otherwise, if $lb^{\psi}$ equals $R_{{\bar y}^{\psi}}$, the algorithm stops
(see Step \rom{4} of Algorithm \ref{Benders01}). As proved in \cite{Assuncao16b}, this stopping condition is
always satisfied in a finite number of iterations.

The algorithm detailed above is applied to obtain optimal solutions for the interval robust-hard problems considered in this work,
and their corresponding bounds are used to evaluate the quality of the solutions produced by the heuristic proposed in the next section.
We highlight that other exact algorithms can be adapted and used for this purpose, such as the \emph{branch-and-cut} algorithm proposed by Pereira and Averbakh \cite{Averbakh13}
for the interval data min-max regret set covering problem. In this work, we chose the logic-based Benders' algorithm, as it has already been successfully applied to solve a wide range of
interval min-max regret problems \cite{Mo06,MontemanniBarta07,Pereira11,MonteGa04c,Averbakh13}.
\section{An LP based heuristic framework for interval min-max regret problems}
\label{s_heuristic}
In this section, we present an LP based heuristic framework for interval robust-hard problems which is applicable to interval min-max regret problems in general.
Consider an interval min-max regret problem $\tilde{\mathcal{R}}$, as defined by (\ref{r00}) and (\ref{r01}) in Section~\ref{s_general_robust_pb}. 
For a given $\bar{y} \in \Omega$, the inner minimization in (\ref{r00}) is a 0-1 ILP problem, namely $\mathcal{G}({s({\bar y}}))$.
Particularly, it corresponds to problem $\mathcal{G}$, defined by (\ref{bilp00})-(\ref{bilp02}), under the cost vector $c^{s({\bar y})}$,
where $s({\bar y})$ is the scenario induced by $\bar{y}$.

Relaxing the integrality on the $x$ variables of $\mathcal{G}({s({\bar y})})$, we obtain the LP problem
\begin{eqnarray}
  \mbox{$\vartheta{({\bar y})}$} = \min && c^{s({\bar y})}x \label{g000}\\
  s.t.  && Ax \geq b, \label{g01}\\
	&& Ix \leq \textbf{1}, \label{g02}\\
	&& x \geq \textbf{0}, \label{g03}
\end{eqnarray}

\noindent whose corresponding dual problem is given by 
\begin{eqnarray}
	\mbox{$({\mathcal{D}}({\bar y}))\quad\vartheta{({\bar y})}$} = \max && ({b^T\lambda} + \textbf{1}^T\mu) \label{gdual00}\\
  s.t.  && A^{T}\lambda + I^{T}\mu \leq ({c^{s({\bar y})}})^T, \label{g04}\\
	&& \lambda \geq \textbf{0}, \label{g05}\\
	&& \mu \leq \textbf{0}. \label{g06}
\end{eqnarray}

\noindent Here, $I$ is the identidy matrix, and the dual variables $\lambda$ and $\mu$ are associated, respectively, with constraints (\ref{g01}) and (\ref{g02}) of the primal problem.
Notice that, in both LP problems, $\vartheta({\bar y})$ corresponds to an optimal solution cost and gives a lower (dual) bound on the optimal value
of $\mathcal{G}(s({\bar y}))$. Precisely,
\begin{equation}
 \vartheta({\bar y}) \leq \overbrace{\min\limits_{x \in \Omega}{c^{s({\bar y})}x}}^{\mathcal{G}(s({\bar y}))}\quad \forall\, \bar{y} \in \Omega. \label{g07}
\end{equation}

Replacing $\mathcal{G}(s({y}))$ by ${\mathcal{D}}({y})$ in (\ref{r00}), we obtain
\begin{eqnarray}
  \mbox{}\min && \Big({ c^{s(y)}y -  \max\;({b^T\lambda} + \textbf{1}^T\mu)}\Big)\\
  s.t. && \mbox{Constraints (\ref{r01}), (\ref{g05}) and (\ref{g06})}, \nonumber \\ 
	&& A^{T}\lambda + I^{T}\mu \leq ({c^{s(y)}})^T \label{g12}.
\end{eqnarray}
\noindent Recall that $c^{s(y)} ={\Big( l_1 + (u_1 - l_1)y_1, \ldots, l_n + (u_n - l_n)y_n \Big)}$ is the cost vector referred to the scenario $s(y)$ induced
by the $y$ variables. Notice that the nested maximization operator can be omitted, giving the following formulation.
\begin{eqnarray}
  \mbox{$(\hat{\mathcal{R}})\quad$}\min && \Big({ c^{s(y)}y - {b^T\lambda} - \textbf{1}^T\mu}\Big) \label{g08}\\
  s.t. && \mbox{Constraints (\ref{r01}), (\ref{g05}), (\ref{g06}) and (\ref{g12})}. \nonumber
\end{eqnarray}

\begin{proposition}
\label{prop01}
 The cost value referred to an optimal solution for $\hat{\mathcal{R}}$ gives an upper bound on the optimal
 solution value of $\tilde{\mathcal{R}}$, and this bound is tight (optimal) if (i) the restriction matrix $A$ is totally unimodular, and (ii) the column vector $b$ is integral.
\end{proposition}

\begin{proof}
 Let $({\tilde y},{\tilde x})$ and $(\hat{y},\hat{\lambda},\hat{\mu})$ be optimal solutions for ${\tilde{\mathcal{R}}}$
 and $\hat{\mathcal{R}}$, respectively.
 The value of the objective function in (\ref{r00}) referred to $({\tilde y},{\tilde x})$ is given by
 $(c^{{s({{\tilde y}})}}{{\tilde y}} - c^{{s({\tilde y})}}\tilde{x})$, whereas that of (\ref{g08}) referred to
 $(\hat{y},\hat{\lambda},\hat{\mu})$ is 
 $(c^{s({\hat y})}{\hat y} - ({b^T\hat{\lambda}} + \textbf{1}^T\hat{\mu}))$. Since $({\tilde y},{\tilde x})$ is optimal for
 $\tilde{\mathcal{R}}$, we have that $(c^{s({{\tilde y}})}{\tilde y} - c^{s({{\tilde y}})}\tilde{x}) \leq (c^{s(y)}{y} - \min\limits_{x \in \Omega}{c^{s(y)}x})$ for all $y \in \Omega$.
In particular, as $\hat{y} \in \Omega$, it holds that 
\begin{equation}
 c^{s({{\tilde y}})}{\tilde y} - c^{s({{\tilde y}})}\tilde{x} \leq c^{s({\hat y})}{\hat{y}} - \min\limits_{x \in \Omega}{c^{s({\hat y})}x} \label{gi}.
\end{equation}
%

As $(\hat{y},\hat{\lambda},\hat{\mu})$ is optimal for $\hat{\mathcal{R}}$, it follows, from (\ref{g000})-(\ref{g07}), that 
\begin{equation}
 c^{s({\hat y})}{\hat{y}} - \min\limits_{x \in \Omega}{c^{s({\hat y})}x} \leq
 c^{s({\hat y})}{\hat{y}} - \vartheta({\hat y}) =
 c^{s({\hat y})}{\hat y} - ({b^T\hat{\lambda}} + \textbf{1}^T\hat{\mu}). \label{g09}
\end{equation}

Then, from (\ref{gi}) and (\ref{g09}), we obtain that
\begin{equation}
c^{s({{\tilde y}})}{\tilde y} - c^{s({{\tilde y}})}\tilde{x} \leq c^{s({\hat y})}{\hat y} - ({b^T\hat{\lambda}} + \textbf{1}^T\hat{\mu}).\label{gi2}
\end{equation}

Now, also suppose that (i) the matrix $A$ is totally unimodular, and (ii) the column vector $b$ is integral.
As $(\hat{y},\hat{\lambda},\hat{\mu})$ is optimal for $\hat{\mathcal{R}}$, it follows, from (\ref{g000})-(\ref{g06}), that 
\begin{equation}
c^{s({\hat y})}{\hat y} - ({b^T\hat{\lambda}} + \textbf{1}^T\hat{\mu}) = 
 c^{s({\hat y})}{\hat{y}} - \vartheta({{\hat y}}) \leq
 c^{s({y})}{y} - \vartheta({y}) \quad \forall y \in \Omega.
\end{equation}

\noindent Particularly, as $\tilde{y} \in \Omega$,
\begin{equation}
 c^{s({\hat y})}{\hat y} - ({b^T\hat{\lambda}} + \textbf{1}^T\hat{\mu}) \leq
 c^{s({\tilde y})}{\tilde y} - \vartheta({\tilde y}). \label{g10}
\end{equation}

Additionally, from assumption (i), it follows that the restriction matrix $(A,I)^T$ referred to (\ref{gdual00})-(\ref{g06}) is also totally unimodular.
Thus, from assumption (ii), inequation (\ref{g07}) is tight, i.e.,
\begin{equation}
\vartheta({\bar y}) = \min\limits_{x \in \Omega}{c^{s({\bar y})}x}\quad \forall {\bar y}\in \Omega. \label{g11}
\end{equation}

\noindent As $({\tilde y},{\tilde x})$ is optimal for $\tilde{\mathcal{R}}$, we obtain, from (\ref{g10}) and (\ref{g11}),
\begin{equation}
 c^{s({\hat y})}{\hat y} - ({b^T\hat{\lambda}} + \textbf{1}^T\hat{\mu})\leq
 c^{s({\tilde y})}{\tilde y} - \vartheta({\tilde y}) = 
 c^{s({\tilde y})}{\tilde y} - \min\limits_{x \in \Omega}{c^{s({\tilde y})}x} = 
 c^{s({\tilde y})}{\tilde y} - {c^{s({\tilde y})}{\tilde x}},
\end{equation}

\noindent which, along with (\ref{gi2}), implies $(c^{s({\hat y})}{\hat y} - ({b^T\hat{\lambda}} + \textbf{1}^T\hat{\mu})) = 
(c^{s({\tilde y})}{\tilde y} - {c^{s({\tilde y})}{\tilde x}})$.
\end{proof}

For a given problem $\mathcal{R}$, defined by (\ref{g00}) and (\ref{d00}), the heuristic framework consists in (\rom{1}) solving the corresponding formulation $\hat{\mathcal{R}}$,
obtaining a solution $(\hat{y},\hat{\lambda},\hat{\mu})$, and (\rom{2}) computing the robustness cost (maximum regret) referred to $\hat{y}$,
considering Theorem~\ref{teo01}.
One may note that the heuristic is applicable not only to interval robust-hard problems, but to any interval min-max regret problem of the
general form of $\mathcal{R}$. From Proposition~\ref{prop01}, whenever the classical optimization problem counterpart can be modeled as
a 0-1 ILP of the form of $\mathcal{G}$, with a totally unimodular restriction matrix and $b$ integral, there is a guarantee of optimality at solving
$\hat{\mathcal{R}}$.
In fact, applying the framework detailed above to the interval data min-max regret versions of the polynomially solvable shortest path and
assignment problems presented,
respectively, in \cite{Karasan01} and \cite{Kasperski05}, leads to the same compact MILP formulations proposed and computationally tested in these works.

Notice that, whenever $\mathcal{G}$ is compact, the resulting formulation $\hat{\mathcal{R}}$ is also compact.
We also highlight that, although the heuristic framework was detailed by the assumption of $\mathcal{G}$ being a minimization problem,
the results also hold for maximization problems, with minor modifications.
In addition, we conjecture that the heuristic framework also provides valid bounds for the wider class of interval data min-max regret problems with compact constraint sets
addressed in \cite{Conde2010,Conde2012}. We believe that Theorem~\ref{teo01} can be extended to this class of problems through the definition of
a specific scenario similar to the one induced by a solution. However, a more careful study needs to be conducted to close this issue.

In this work, we do not present any theoretical guarantee of quality for the solutions obtained by the heuristic, as it
may rely on specific problem dependent factors, such as the strength of the formulation adopted to model the classical counterpart
problem $\mathcal{G}$. Nevertheless, in the next section, we present two very successful applications of the heuristic in solving
interval robust-hard problems.
\section{Case studies in solving interval robust-hard problems}
\label{s_case_studies}
In this section, we define the interval robust-hard problems used as case studies for the proposed heuristic framework. For each problem, we give a
mathematical formulation according to the modeling technique presented in Section~\ref{s_general_robust_pb} and devise the corresponding MILP
formulation tackled by the heuristic. For the interval data min-max regret restricted shortest path problem, which is introduced in this work,
we give a more detailed description, along with a brief motivation and some related works. For simplicity, let \emph{robust} stand for \emph{interval data min-max regret} in the designation of each problem. 
\subsection{The Restricted Robust Shortest Path problem (R-RSP)}
\label{s_description}
The Restricted Robust Shortest Path problem (R-RSP) is an interval data min-max regret
version of the Restricted Shortest Path problem (R-SP),
an extensively studied NP-hard problem \cite{Garey79, Handler80,
  Aneja83, Beasley89, Hassin92, Wang96}. Consider a digraph $G=(V,A)$,
where $V$ is the set of vertices, and $A$ is the set of arcs. With each
arc $(i,j) \in A$, we associate a resource consumption
$d_{ij} \in \mathbb{Z}_{+}$ and a continuous cost interval
[${l}_{ij}$,${u}_{ij}$], where ${l}_{ij} \in \mathbb{Z}_{+}$ is the
lower bound, and ${u}_{ij}\in \mathbb{Z_{+}}$ is the upper bound on
this interval of cost, with ${l}_{ij} \leq {u}_{ij}$.  An origin
vertex $o \in V$ and a destination one $t \in V$ are also given,
as well as a value $\beta \in \mathbb{Z}$, parameter used to
limit the resource consumed along a path from $o$ to $t$ in $G$, as discussed in the sequel.
An example of an R-RSP instance is given in
Figure \ref{fig_example}.
\begin{figure}[!h]
\center
\includegraphics*[scale=0.8]{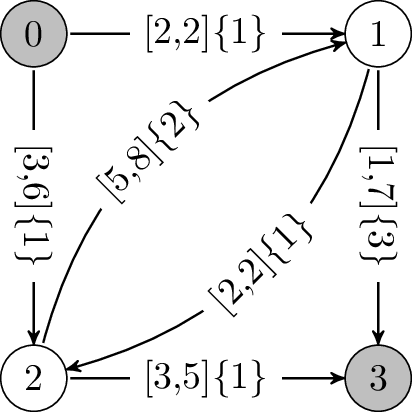}
\caption{Example of an R-RSP instance, with origin $o = 0$ and
  destination $t=3$. The notation $[{l}_{ij},{u}_{ij}]\{d_{ij}\}$
  means, respectively, the cost interval $[{l}_{ij},{u}_{ij}]$ and the
  resource consumption $\{d_{ij}\}$ associated with each arc $(i,j)
  \in A$.}
\label{fig_example}
\end{figure}

Here, a scenario $s$ is an assignment of arc costs, where a cost
$c_{ij}^s \in [{l}_{ij},{u}_{ij}]$ is fixed for all $(i,j) \in A$.
Let $\mathcal{P}$ be the set of all paths from $o$ to $t$ and $A[p]$
be the set of the arcs that compose a path $p \in \mathcal{P}$.
Also let $\mathcal{S}$ be the set of all possible scenarios of $G$. The
cost of a path $p \in \mathcal{P}$ in a scenario $s \in \mathcal{S}$
is given by $C_p^s = \sum\limits_{(i,j) \in A[p]}{c_{ij}^s}$.
Similarly, the resource consumption referred to a path $p \in
\mathcal{P}$ is given by $D_p = \sum\limits_{(i,j) \in
  A[p]}{d_{ij}}$. Also consider $\mathcal{P}(\beta) = \{p \in
\mathcal{P}$ $|$ $D_p \leq \beta\}$, the subset of paths in
$\mathcal{P}$ whose resource consumptions are smaller than or equal to
$\beta$.
\begin{definition}
  A path $p^*(s,\beta) \in \mathcal{P}(\beta)$ is said to be a
  $\beta$-\emph{restricted shortest path} in a scenario $s \in
  \mathcal{S}$ if it has the smallest cost in $s$ among all paths in
  $\mathcal{P}(\beta)$, i.e., $p^*(s,\beta) = \arg\min\limits_{p \in
    \mathcal{P}(\beta)}{C_{p}^s}$.
\end{definition}
\begin{definition}
  The \emph{regret (robust deviation)} of a path $p \in \mathcal{P}(\beta)$ in
  a scenario $s \in \mathcal{S}$, denoted by $r^{(s,\,\beta)}_p$, is the
  difference between the cost $C^s_p$ of $p$ in $s$ and the cost of a
  $\beta$-restricted shortest path $p^*(s,\beta) \in
  \mathcal{P}(\beta)$ in $s$, i.e., $r^{(s,\,\beta)}_p = C^s_p -
  C^s_{p^*(s,\,\beta)}$.
\end{definition}
\begin{definition}
  The $\beta$-\emph{restricted robustness cost} of a path $p \in
  \mathcal{P}(\beta)$, denoted by $R_p^{\beta}$, is the maximum regret of $p$ among all possible scenarios, i.e., $R_p^{\beta} =
  \max\limits_{s \in \mathcal{S}}{r^{(s,\,\beta)}_p} $.
\end{definition}
\begin{definition}
  A path $p^* \in \mathcal{P}(\beta)$ is said to be a
  $\beta$-\emph{restricted robust path} if it has the smallest
  $\beta$-restricted robustness cost among all paths in
  $\mathcal{P}(\beta)$, i.e., $p^* = \arg\min\limits_{p \in
    \mathcal{P}(\beta)}{R_{p}^{\beta}}$.
\end{definition}
\begin{definition}
  \emph{R-RSP} consists in finding a $\beta$-restricted robust path $p^* \in \mathcal{P}(\beta)$.
\end{definition}
R-RSP has applications in determining paths in urban areas, where
travel distances are known in advance, but travel times may vary
according to unpredictable traffic jams, bad weather conditions, etc.
Here, uncertainties are represented by values in continuous intervals,
which estimate the minimum and the maximum traveling times to traverse
each pathway. For instance, R-RSP can model situations involving
electrical vehicles with a limited battery (energy) autonomy, when one
wants to find the fastest robust path with a length (distance)
constraint. A similar application arises in telecommunications
networks, when one wants to determine a path to efficiently send a data
packet from an origin node to a destination one in a network. Due to the
varying traffic load, transmission links are subject to uncertain
delays. Moreover, each link is associated with a packet loss rate. In
order to guarantee Quality of Service (QoS)~\cite{Wang96,
Apostolopoulos98}, a limit is imposed on the total packet loss rates
of the path used.

The classical R-SP is a particular case of R-RSP in which $l_{ij} = u_{ij}$ $\forall\, (i,j) \in A$.
As R-SP is known to be NP-hard \cite{Garey79,Handler80}
even for acyclic graphs \cite{Wang96}, the same holds for R-RSP.
The main exact algorithms to solve R-SP can be subdivided into two groups:
\emph{Lagrangian relaxation} and \emph{dynamic programming} procedures.
The former procedures use Lagrangian relaxation to handle ILP formulations
for the problem (see, e.g., \cite{Handler80,Beasley89}).
In addition, preprocessing techniques have
been presented in \cite{Aneja83} and refined in
\cite{Beasley89}. These techniques identify arcs and vertices that
cannot compose an optimal solution for R-SP through the analysis of
the reduced costs related to the resolution of dual Lagrangian
relaxations. More recently, \cite{Santos07} proposed a path ranking
approach that linearly combines the arc costs and the resource
consumption values to generate a descent direction of search.
In turn, dynamic programming procedures for R-SP consist
of label-setting and label-correcting algorithms, such as the one
proposed in \cite{Joksch66} and further improved in
\cite{Dumitrescu03} by the addition of preprocessing
strategies. Recently, Zhu and Wilhelm~\cite{Zhu2012} developed a
three-stage label-setting algorithm that runs in pseudo-polynomial time
and figures among the state-of-the-art methods to solve R-SP, along with the
algorithms presented in \cite{Santos07} and \cite{Dumitrescu03}.  We
refer to \cite{Pugliese13} for a survey on exact methods to solve
R-SP.

Although NP-hard, R-SP can be solved efficiently by some of the aforementioned procedures
(particularly, the ones proposed in \cite{Santos07,Dumitrescu03,Zhu2012}). Moreover,
optimization softwares, as CPLEX\footnote{http://www-01.ibm.com/software/commerce/optimization/cplex-optimizer/}, are
also competitive in handling reasonably-sized instances (with up to 3000 vertices) of the problem \cite{Zhu2012}.

R-RSP is also a generalization of the interval data min-max regret Robust
Shortest Path problem (RSP) \cite{Kouvelis1997, Averbakh05, Rosenhead72, Pawel04}.
RSP consists in finding a robust path (from
the origin vertex to the destination one) considering the min-max
regret criterion, with no additional resource consumption restriction on the solution path.
Therefore, RSP can be reduced to R-RSP by considering $\beta=0$ and $d_{ij} = 0$ $\forall\, (i,j) \in A$. 

Preprocessing techniques able to identify
arcs that cannot compose an optimal solution for RSP have been proposed in
\cite{Karasan01} and later improved in \cite{Catanzaro11}.  A compact
MILP formulation for the problem, based on the dual of an LP formulation for
the classical shortest path problem, was presented in \cite{Karasan01}. Exact algorithms have been proposed for RSP,
such as the \emph{branch-and-bound} algorithm of \cite{MonteGa04b} and the
logic-based Benders' algorithm of \cite{MonteGa04c}, which is able to solve instances with up to 4000 vertices.
Moreover, the FPTAS of \cite{Kasperski07} for interval min-max regret problems is applicable to RSP when the problem is considered in series-parallel graphs.
However, as pointed out by the end of Section~\ref{s_rw}, this FPTAS does not naturally extend to interval
robust-hard problems, such as R-RSP.
\subsubsection{Mathematical formulation}
\label{s_model}
The mathematical formulation here presented makes use of Theorem~\ref{teo01}, which can be stated for the specific case of R-RSP as follows.

\begin{proposition}
\label{prop_rrsp_max}
Given a value $\beta \in \mathbb{Z}$ and a path $p \in \mathcal{P}(\beta)$, the regret of $p$ is maximum in the
scenario $s(p)$ induced by $p$, where the costs of all the arcs of $p$ are in their corresponding upper bounds and the costs of all the other arcs are in their
lower bounds.
 \end{proposition}
\begin{figure}[!h]
\center
\includegraphics*[scale=0.8]{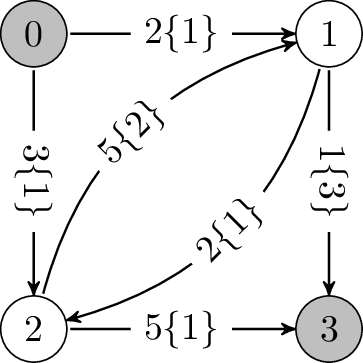}
\caption{Scenario $s_{\tilde{p}}$ induced by the path $\tilde{p}=\{0,1,2,3\}$ in the digraph presented in Figure \ref{fig_example}.
For each arc $(i,j) \in A$,
the notation ${c}_{ij}^{s({\tilde{p}})}\{d_{ij}\}$ means, respectively, the arc cost ${c}_{ij}^{s({\tilde{p}})}$
in the scenario $s({\tilde{p}})$ and its resource consumption $\{d_{ij}\}$.}
\label{fig_result}
\end{figure}

Consider the digraph presented in Figure~\ref{fig_example} and the scenario $s({\tilde{p}})$ induced by the path $\tilde{p} = \{0,1,2,3\}$ (as showed in Figure~\ref{fig_result}). Also let the resource limit be $\beta = 3$.
According to Proposition \ref{prop_rrsp_max}, the $\beta$-restricted robustness cost of $\tilde{p}$ is given by
$R_{\tilde{p}}^\beta = r^{(s({\tilde{p}}),\,\beta)}_{\tilde{p}} = C^{s({\tilde{p}})}_{\tilde{p}} - C^{s({\tilde{p}})}_{{p}^*(s({\tilde{p}}),\,\beta)} = (2+2+5) - (3+5)=1$.
Proposition \ref{prop_rrsp_max} reduces R-RSP to finding a path $p^* \in
\mathcal{P}(\beta)$ such that $p^* = \arg\min\limits_{p \in
  \mathcal{P}}{r^{(s({p}),\,\beta)}_{p}}$, i.e., $p^* =
\arg\min\limits_{p \in \mathcal{P}(\beta)}{\{C^{s({p})}_{p} -
  C^{s({p})}_{{p}^*(s({p}),\,\beta)}\}}$.
Nevertheless, computing the $\beta$-restricted robustness cost of any feasible
solution $p$ for R-RSP still implies finding a $\beta$-restricted shortest
path $p^*(s(p),\beta)$ in the scenario $s({p})$ induced by $p$, and this problem is NP-hard.

Now, let us consider the following result, which helps describing an optimal solution path for R-RSP.
\begin{property}
\label{property01}
Given two arbitrary sets $Z_1$ and $Z_2$, it holds that $Z_1 = (Z_1 \cap Z_2) \cup (Z_1 \backslash Z_2)$.
\end{property}
 
\begin{theorem}
\label{teo02}
 Given a value $\beta \in \mathbb{Z}$ and a non-elementary path $p \in \mathcal{P}(\beta)$,
 for any elementary path $\tilde{p} \in \mathcal{P}(\beta)$ such that $A[\tilde{p}] \subset A[p]$, it holds that
 $r_{\tilde{p}}^{(s({\tilde{p}}),\beta)} \leq r_{p}^{(s(p),\beta)}$.
 \end{theorem}

\begin{proof}
Consider a value $\beta \in \mathbb{Z}$ and a non-elementary path $p \in \mathcal{P}(\beta)$. By definition, $p$ contains at least one cycle.
Let $G[p]$ be the subgraph of $G$ induced by the arcs in $A[p]$ and $\tilde{p}$ be an elementary path from $o$ to $t$ in $G[p]$.
Clearly, $A[\tilde{p}] \subset A[p]$ and, therefore, $\tilde{p} \in \mathcal{P}(\beta)$. By definition,
{
\begin{eqnarray}
  r_{\tilde{p}}^{(s({\tilde{p}}),\beta)} = C_{\tilde{p}}^{s({\tilde{p}})} - C_{{p}^*(s({\tilde{p}}),\beta)}^{s({\tilde{p}})}.\label{b0}
\end{eqnarray}}

Consider the set $ {\bar A} = A[p] \backslash A[{\tilde{p}}]$ of the arcs in $p$ which do not belong to $\tilde{p}$.
As $p$ is supposed to contain at least one cycle, then ${\bar A} \neq \emptyset $. Since $A[{\tilde{p}}] \subset A[p]$, the difference
between scenarios $s(p)$ and $s({\tilde{p}})$ consists of the cost values assumed by the arcs in ${\bar A}$. More precisely,
{
\begin{eqnarray}
c_{ij}^{s({\tilde{p}})} = c_{ij}^{s(p)} \hspace{1.5cm} \forall (i,j) \in A \backslash {\bar A},\label{b1.0} \\
c_{ij}^{s(p)} = {u}_{ij} \hspace{2cm} \forall (i,j) \in {\bar A},\label{b1} \\
c_{ij}^{s({\tilde{p}})} = {l}_{ij} \hspace{2.1cm} \forall (i,j) \in {\bar A}.\label{b2}
\end{eqnarray}}

It follows that
{
\begin{eqnarray}
C_{\tilde{p}}^{s({\tilde{p}})} = \sum\limits_{(i,j) \in A[\tilde{p}]} {u}_{ij} = \sum\limits_{(i,j) \in A[p]} {u}_{ij} 
- \sum\limits_{(i,j) \in A[p] \backslash A[{\tilde{p}}]} {u}_{ij} = C_{p}^{s({p})} - \sum\limits_{(i,j) \in {\bar A}} {u}_{ij}. \label{b3}
\end{eqnarray}}

Applying Property~\ref{property01} to the sets $A[{{{p}^*(s({\tilde{p}}),\beta)}}]$ and $A[p]$:
{
\begin{eqnarray}
A[{{{p}^*(s({\tilde{p}}),\beta)}}] = \overbrace{(A[{{{p}^*(s({\tilde{p}}),\beta)}}] \cap A[p])}^{(a_1)} \cup (A[{{{p}^*(s({\tilde{p}}),\beta)}}] \backslash A[p]).\label{b4}
\end{eqnarray}}

Since $A[{\tilde{p}}] \subset A[p]$ and ${\bar A} = A[p] \backslash A[{\tilde{p}}]$, it follows that
{
\begin{eqnarray}
A[p] = A[\tilde{p}] \cup (A[p] \backslash A[{\tilde{p}}]) = A[\tilde{p}] \cup {\bar A}.
\end{eqnarray}}

Therefore, expression $(a_1)$ of ($\ref{b4}$) can be rewritten as

{
\begin{equation*}
    (A[{{{p}^*(s({\tilde{p}}),\beta)}}] \cap A[p]) = (A[{{{p}^*(s({\tilde{p}}),\beta)}}] \cap (A[\tilde{p}] \cup {\bar A})) =
\end{equation*}}
{\small
\begin{equation}
    (A[{{{p}^*(s({\tilde{p}}),\beta)}}] \cap A[{\tilde{p}}]) \cup (A[{{{p}^*(s({\tilde{p}}),\beta)}}] \cap {\bar A}). \label{b5}
\end{equation}}

Applying (\ref{b4}) and (\ref{b5}) to $C_{{p}^*(s({\tilde{p}}),\beta)}^{s({\tilde{p}})}$ and $C_{{p}^*(s({\tilde{p}}),\beta)}^{s_p}$, we obtain:
{
\begin{equation*}
C_{{p}^*(s({\tilde{p}}),\beta)}^{s({\tilde{p}})} = \sum\limits_{(i,j) \in A[{{{p}^*(s({\tilde{p}}),\beta)}}] \cap A[{\tilde{p}}]}c_{ij}^{s({\tilde{p}})} + 
\sum\limits_{(i,j) \in A[{{{p}^*(s({\tilde{p}}),\beta)}}] \cap {\bar A}}c_{ij}^{s({\tilde{p}})}
\end{equation*}
\begin{equation}
+ \sum\limits_{(i,j) \in A[{{{p}^*(s({\tilde{p}}),\beta)}}] \backslash A[p]}c_{ij}^{s({\tilde{p}})}, \label{b6} \\
\end{equation}

\begin{equation*}
C_{{p}^*(s({\tilde{p}}),\beta)}^{s(p)} = \sum\limits_{(i,j) \in A[{{{p}^*(s({\tilde{p}}),\beta)}}] \cap A[{\tilde{p}}]}c_{ij}^{s(p)} + 
\sum\limits_{(i,j) \in A[{{{p}^*(s({\tilde{p}}),\beta)}}] \cap {\bar A}}c_{ij}^{s(p)}
\end{equation*}

\begin{equation}
+ \sum\limits_{(i,j) \in A[{{{p}^*(s({\tilde{p}}),\beta)}}] \backslash A[p]}c_{ij}^{s(p)}. \label{b7}
\end{equation}}

Considering ($\ref{b1.0}$)-($\ref{b2}$), ($\ref{b6}$) and ($\ref{b7}$), we deduce that the difference between the cost of path ${p}^*(s({\tilde{p}}),\beta)$ in $s(p)$ and its cost
in $s({\tilde{p}})$ is given by the arcs which are simultaneously in ${\bar A}$ and in $A[{p}^*(s({\tilde{p}}),\beta)]$.
Thus, expressions (\ref{b6}) and (\ref{b7}) can be reformulated as
{
\begin{equation*}
C_{{p}^*(s({\tilde{p}}),\beta)}^{s({\tilde{p}})} = \sum\limits_{(i,j) \in A[{{{p}^*(s({\tilde{p}}),\beta)}}] \cap A[{\tilde{p}}]}{u}_{ij} + 
\sum\limits_{(i,j) \in A[{{{p}^*(s({\tilde{p}}),\beta)}}] \cap {\bar A}}{l}_{ij}
\end{equation*}
\begin{equation}
+ \sum\limits_{(i,j) \in A[{{{p}^*(s({\tilde{p}}),\beta)}}] \backslash A[p]}{l}_{ij}, \label{b8}
\end{equation}
\begin{equation*}
C_{{p}^*(s({\tilde{p}}),\beta)}^{s(p)} = \sum\limits_{(i,j) \in A[{{{p}^*(s({\tilde{p}}),\beta)}}] \cap A[{\tilde{p}}]}{u}_{ij} + 
\sum\limits_{(i,j) \in A[{{{p}^*(s({\tilde{p}}),\beta)}}] \cap {\bar A}}{u}_{ij}
\end{equation*}
\begin{equation}
+ \sum\limits_{(i,j) \in A[{{{p}^*(s({\tilde{p}}),\beta)}}] \backslash A[p]}{l}_{ij}. \label{b9}
\end{equation}}

Subtracting (\ref{b9}) from (\ref{b8}),
{
\begin{eqnarray}
C_{{p}^*(s({\tilde{p}}),\beta)}^{s({\tilde{p}})} - C_{{p}^*(s({\tilde{p}}),\beta)}^{s(p)} =
\sum\limits_{(i,j) \in A[{{{p}^*(s({\tilde{p}}),\beta)}}] \cap {\bar A}}{l}_{ij} -
\sum\limits_{(i,j) \in A[{{{p}^*(s({\tilde{p}}),\beta)}}] \cap {\bar A}}{u}_{ij}. \label{b10}
\end{eqnarray}}

Therefore,
{
\begin{eqnarray}
C_{{p}^*(s({\tilde{p}}),\beta)}^{s({\tilde{p}})} = C_{{p}^*(s({\tilde{p}}),\beta)}^{s(p)} -
\sum\limits_{(i,j) \in A[{{{p}^*(s({\tilde{p}}),\beta)}}] \cap {\bar A}}({u}_{ij} - {l}_{ij}). \label{b11}
\end{eqnarray}}

Applying (\ref{b3}) and (\ref{b11}) in (\ref{b0}):
{
\begin{equation*}
  r_{\tilde{p}}^{(s({\tilde{p}}),\beta)} = C_{p}^{s({p})} - \sum\limits_{(i,j) \in {\bar A}} {u}_{ij} - \Big( C_{{p}^*(s({\tilde{p}}),\beta)}^{s(p)} -
  \sum\limits_{(i,j) \in A[{{{p}^*(s({\tilde{p}}),\beta)}}] \cap {\bar A}}({u}_{ij} - {l}_{ij})\Big) = 
\end{equation*}}
{
\begin{equation}
  C_{p}^{s({p})} - C_{{p}^*(s({\tilde{p}}),\beta)}^{s(p)} +
  \sum\limits_{(i,j) \in A[{{{p}^*(s({\tilde{p}}),\beta)}}] \cap {\bar A}}({u}_{ij} - {l}_{ij}) -
  \sum\limits_{(i,j) \in {\bar A}} {u}_{ij}. \label{b12}
\end{equation}}

One may note that
{
\begin{equation}
\sum\limits_{(i,j) \in A[{{{p}^*(s({\tilde{p}}),\beta)}}] \cap {\bar A}}({u}_{ij} - {l}_{ij}) -
  \sum\limits_{(i,j) \in {\bar A}} {u}_{ij} \leq
%
\sum\limits_{(i,j) \in {\bar A}}({u}_{ij} - {l}_{ij}) -
  \sum\limits_{(i,j) \in {\bar A}} {u}_{ij} \leq
\sum\limits_{(i,j) \in {\bar A}}(-{l}_{ij}). \label{b13}
\end{equation}}

Since ${\bar A} \subset A$, and ${l}_{ij} \geq 0$ for all $(i,j) \in A$, it follows that $\sum\limits_{(i,j) \in {\bar A}}(- {l}_{ij}) \leq 0$ 
and, thus,
{
\begin{eqnarray}
\sum\limits_{(i,j) \in A[{{{p}^*(s({\tilde{p}}),\beta)}}] \cap {\bar A}}({u}_{ij} - {l}_{ij}) -
  \sum\limits_{(i,j) \in {\bar A}} {u}_{ij} \leq 0. \label{b14}
\end{eqnarray}}
 
From (\ref{b12}) and (\ref{b14}), 
{
\begin{eqnarray}
r_{\tilde{p}}^{(s({\tilde{p}}),\beta)} \leq C_{p}^{s({p})} - C_{{p}^*(s({\tilde{p}}),\beta)}^{s(p)}. \label{b15}
\end{eqnarray}}

As ${p}^*(s(p),\beta)$ is a path with the smallest cost in $s(p)$ among all the paths in $\mathcal{P}(\beta)$, including
${p}^*(s({\tilde{p}}),\beta)$, it holds that $C_{{p}^*(s({\tilde{p}}),\beta)}^{s(p)} \geq C_{{p}^*(s(p),\beta)}^{s(p)}$. Thus,
{
\begin{eqnarray}
r_{\tilde{p}}^{(s({\tilde{p}}),\beta)} \leq C_{p}^{s({p})} - C_{{p}^*(s({\tilde{p}}),\beta)}^{s(p)} \leq C_{p}^{s({p})} - C_{{p}^*(s(p),\beta)}^{s(p)}.\label{b16}
\end{eqnarray} }

By definition, $r_p^{(s(p),\beta)} = C_{p}^{s({p})} - C_{{p}^*(s(p),\beta)}^{s(p)}$. Therefore,
$r_{\tilde{p}}^{(s({\tilde{p}}),\beta)} \leq r_p^{(s(p),\beta)}$.
\end{proof}

Now, we can devise a mathematical formulation for R-RSP from the generic model $\tilde{\mathcal{R}}$, defined by (\ref{r00}) and (\ref{r01}).
Consider the decision variables $y$ on the choice of arcs belonging or not to a $\beta$-restricted
robust path: $y_{ij} = 1$ if the arc $(i,j) \in A$ belongs to the solution path; $y_{ij} = 0$, otherwise. Likewise, let the binary variables
$x$ identify a $\beta$-restricted shortest path in the scenario induced by the path defined by $y$, such that $x_{ij} = 1$ if the arc $(i,j) \in A$
belongs to this $\beta$-restricted shortest path, and $x_{ij} = 0$, otherwise. A nonlinear compact formulation for R-RSP is given by
\begin{eqnarray}
 \min \limits_{y \in \mathcal{P}(\beta)}{\Bigg( \sum \limits_{(i,j) \in A}{u_{ij}y_{ij}} - \min \limits_{x \in \mathcal{P}(\beta)}{
 \sum \limits_{(i,j) \in A}{(l_{ij} + (u_{ij} - l_{ij})y_{ij})x_{ij}}} \Bigg)}. \label{f00}
\end{eqnarray}

As discussed in Section~\ref{s_general_robust_pb}, we can derive an MILP formulation from (\ref{f00}) by adding a free variable $\rho$ and linear
constraints that explicitly bound $\rho$ with respect to all the feasible paths that $x$ can represent. The resulting formulation
is provided from \eqref{f01} to \eqref{f08}.
\begin{eqnarray}
  \min && \sum \limits_{(i,j) \in A}{u_{ij}y_{ij}} - \rho \label{f01} \\
  s.t. && \sum \limits_{j:(j,o) \in A}{y_{jo}} - \sum \limits_{k:(o,k) \in A}{y_{ok}} = -1, \label{f03} \\
        && \sum \limits_{j:(j,i) \in A}{y_{ji}} - \sum \limits_{k:(i,k) \in A}{y_{ik}} = 0 \quad \forall\, i \in V \backslash \{o,t\}, \label{f05} \\
        && \sum \limits_{j:(j,t) \in A}{y_{jt}} - \sum \limits_{k:(t,k) \in A}{y_{tk}} = 1, \label{f04} \\
        && \sum \limits_{(i,j) \in A}{d_{ij}y_{ij}} \leq \beta, \label{f06} \\
        && \rho \leq \sum \limits_{(i,j) \in A}{( l_{ij} + (u_{ij} - l_{ij})y_{ij})\bar{x}_{ij}} \quad \forall\, \bar{x} \in \mathcal{P}(\beta), \label{f02} \\
        && y_{ij} \in \{0,1\} \quad \forall\, (i,j) \in A,   \label{f07} \\
        && \rho \mbox{ free}. \label{f08}
\end{eqnarray}

The flow conservation constraints (\ref{f03})-(\ref{f04}), along with the domain constraints (\ref{f07}),
ensure that the $y$ variables define a path from the origin to the destination vertices. In fact, as pointed out in \cite{Karasan01},
these constraints do not prevent the existence of additional cycles of cost zero disjoint from the solution path. Notice,
however, that every arc $(i,j)$ of these cycles necessarily has $l_{ij} = u_{ij} = 0$ and, thus, they do not modify
the optimal solution value. Hence, these cycles are not taken into account hereafter.

Constraint (\ref{f06}) limits the resource consumption of the path defined by $y$ to be at most $\beta$, whereas constraints (\ref{f02}) guarantee
that $\rho$ does not exceed the value related to the inner minimization in (\ref{f00}). Note that, in (\ref{f02}), $\bar{x}$ is a constant vector,
one for each path in $\mathcal{P}(\beta)$.  Moreover, these constraints are tight whenever $\bar{x}$ identifies a $\beta$-restricted shortest path in
the scenario induced by the path that $y$ defines. Constraint (\ref{f08}) gives the domain of the variable $\rho$.

Notice that, in the definition of R-RSP, we do not impose that a vertex in the solution path must be traversed at most once. However, if that is the case,
Theorem~\ref{teo02} indicates that the formulation presented above can also be used to determine an elementary solution path for R-RSP by
simply discarding some edges from the cycles that may appear in the solution. In fact, Proposition~\ref{prop_rrsp_max} and Theorem~\ref{teo02} imply that,
for any $\beta \in \mathbb{Z}$, if $\mathcal{P}(\beta) \neq \emptyset$, then there is an elementary path $p \in \mathcal{P}(\beta)$ which is a $\beta$-restricted robust path.

In this study, we use the logic-based Benders' algorithm discussed in the Section~\ref{s_benders} to solve the formulation detailed above.
\subsubsection{An LP based heuristic for R-RSP}
In this section, we apply to R-RSP the heuristic framework presented in Section~\ref{s_heuristic}.
To this end, consider the following R-SP ILP formulation used to compute a $\beta$-restricted shortest path
$p^*(s,\beta) \in \mathcal{P}(\beta)$ in a scenario $s \in \mathcal{S}$. The binary variables $x$ define $p^*(s,\beta)$, such that $x_{ij} = 1$ if the arc $(i,j) \in A$ belongs to $A[p^*(s,\beta)]$, and $x_{ij} = 0$, otherwise.
\begin{eqnarray}
  \mbox{$(\mathcal{I}_1)\quad$}\min & & \sum \limits_{(i,j) \in A}{c_{ij}^sx_{ij}} \label{r-sp0} \\
  s.t. && \sum \limits_{j:(j,o) \in A}{x_{jo}} - \sum \limits_{k:(o,k) \in A}{x_{ok}} = -1, \label{r-sp1} \\
        && \sum \limits_{j:(j,i) \in A}{x_{ji}} - \sum \limits_{k:(i,k) \in A}{x_{ik}} = 0 \quad \forall\, i \in V \backslash \{o,t\}, \label{r-sp3} \\
        && \sum \limits_{j:(j,t) \in A}{x_{jt}} - \sum \limits_{k:(t,k) \in A}{x_{tk}} = 1, \label{r-sp2} \\
        && \sum \limits_{(i,j) \in A}{d_{ij}x_{ij}} \leq \beta, \label{r-sp4} \\
        && x_{ij} \in \{0,1\} \quad \forall\, (i,j) \in A. \label{r-sp5}
\end{eqnarray}

The objective function in (\ref{r-sp0}) represents the cost, in the scenario $s$, of the path defined by $x$, while constraints (\ref{r-sp1})-(\ref{r-sp5}) ensure that $x$ identifies a path in $\mathcal{P}(\beta)$. Relaxing the integrality on $x$, we obtain the following LP formulation:
\begin{eqnarray}
  \mbox{$(\mathcal{L}_1)\quad$}\theta^{(s,\,\beta)} = \min && \sum \limits_{(i,j) \in A}{c_{ij}^sx_{ij}} \label{r-spr0} \\
  s.t. && \mbox{Constraints (\ref{r-sp1})-(\ref{r-sp4}),} \nonumber \\
        && x_{ij} \geq 0 \quad \forall\, (i,j) \in A. \label{r-spr5}
\end{eqnarray}

The domain constraints $x_{ij} \leq 1$ for all $(i,j) \in A$ were omitted from $\mathcal{L}_1$, since they are redundant.
Let $\theta^{(s,\,\beta)}$ be the optimal value for problem $\mathcal{L}_1$ in a scenario $s$. Observe that $\theta^{(s,\,\beta)}$ provides a lower bound on
the solution of $\mathcal{I}_1$ in the scenario $s$. For the sake of clarity, let us define a new metric to evaluate the quality of a path in $\mathcal{P}(\beta)$.
\begin{definition}
The $\beta$-\emph{heuristic robustness cost} of a path $p \in \mathcal{P}(\beta)$, denoted by $H_p^{\beta}$, is the difference between the cost
$C_p^{s(p)}$ of $p$ in the scenario $s(p)$ induced by $p$ and the relaxed cost $\theta^{(s(p),\,\beta)}$ in $s(p)$, i.e., $H_p^{\beta} = C_p^{s(p)} -
\theta^{(s(p),\,\beta)}$.
\end{definition}
\begin{definition}
A path $\tilde{p}^* \in \mathcal{P}(\beta)$ is said to be a $\beta$-\emph{heuristic robust path} if it has the smallest $\beta$-heuristic robustness cost among all the paths in $\mathcal{P}(\beta)$, i.e., $\tilde{p}^* = \arg\min\limits_{p \in \mathcal{P}(\beta)}{H_{p}^{\beta}}$.
\end{definition}
In the case of R-RSP, the proposed heuristic aims at finding a $\beta$-heuristic robust path and relies on the hypothesis that such a path is a near-optimal
solution for R-RSP. The problem of finding a $\beta$-heuristic robust path can be modeled by adapting formulation (\ref{f00}).
To this end, the binary variables $y$ now represent a $\beta$-heuristic robust path in $\mathcal{P}(\beta)$. Furthermore, considering the scenario $s(y)$ induced by the path defined by $y$,
the nested minimization in (\ref{f00}) is replaced by $\theta^{(s({y}),\,\beta)}$. We obtain:
\begin{eqnarray}
 \min \limits_{y \in \mathcal{P}(\beta)} && {\Bigg( \sum \limits_{(i,j) \in A}{u_{ij}y_{ij}} - \theta^{(s(y),\,\beta)} \Bigg)}. \label{f00-relax}
\end{eqnarray}

Given a scenario $s \in \mathcal{S}$, the optimal value assumed by $\theta^{(s,\,\beta)}$ can be represented by the dual of $\mathcal{L}_1$, as follows:
\begin{eqnarray}
  \mbox{$(\tilde{\mathcal{L}_1})\quad$}\theta^{(s,\beta)} = \max & & (\lambda_{t} - \lambda_{o} - \beta\mu) \label{r-sprd0} \\
  s.t. && \lambda _j \leq \lambda_i + c_{ij}^s + d_{ij}\mu \quad \forall\, (i,j) \in A, \label{r-sprd1} \\
        && \mu \geq 0, \label{r-sprd2} \\
        && \lambda_k \mbox{ free} \quad \forall\, k \in V. \label{r-sprd3}
\end{eqnarray}

The dual variables $\{\lambda_k: k \in V\}$ and $\mu$ are associated, respectively, with constraints (\ref{r-sp1})-(\ref{r-sp2}) and with constraint (\ref{r-sp4}) in the primal problem $\mathcal{L}_1$. Since $\tilde{\mathcal{L}_1}$ is a maximization problem, its objective function, along with (\ref{r-sprd1})-(\ref{r-sprd3}), can be used to replace the relaxed cost $\theta^{(s(y),\beta)}$ in (\ref{f00-relax}), thus deriving the following formulation:
\begin{eqnarray}
 \min \limits_{y \in \mathcal{P}(\beta)} && {\Bigg( \sum \limits_{(i,j) \in A}{u_{ij}y_{ij}} -
 \overbrace{(\lambda_{t} - \lambda_{o} - \beta\mu)}^{\mbox{From ($\ref{r-sprd0}$)}} \Bigg)} \label{r-rspr0} \\
 s.t. && \lambda _j \leq \lambda_i + l_{ij} + (u_{ij}-l_{ij})y_{ij} + d_{ij}\mu \quad \forall\, (i,j) \in A, \label{r-rspr1} \\
        && \mu \geq 0, \label{r-rspr2} \\
        && \lambda_k \mbox{ free} \quad \forall\, k \in V. \label{r-rspr3}
\end{eqnarray}

Notice that constraints (\ref{r-rspr1}) consider the cost of each arc $(i,j) \in A$ in the scenario $s(y)$ induced by the path identified by the $y$ variables, i.e., the cost of each arc $(i,j) \in A$ is given by $(l_{ij} + (u_{ij}-l_{ij})y_{ij})$. The domain constraints (\ref{r-rspr2}) and (\ref{r-rspr3}) related to $\tilde{\mathcal{L}_1}$ remain the same. Now, we give an MILP formulation for the problem of finding a $\beta$-heuristic robust path.
\begin{eqnarray}
  \mbox{$(\mathcal{H}_1)\quad$}\min & & \Bigg(\sum \limits_{(i,j) \in A}{u_{ij}y_{ij}} - \lambda_{t} + \lambda_{o} + \beta\mu \Bigg)\label{r-rsph0} \\
  s.t. && \sum \limits_{j:(j,o) \in A}{y_{jo}} - \sum \limits_{k:(o,k) \in A}{y_{ok}} = -1, \label{r-rsph2} \\
        && \sum \limits_{j:(j,i) \in A}{y_{ji}} - \sum \limits_{k:(i,k) \in A}{y_{ik}} = 0 \quad \forall\, i \in V \backslash \{o,t\}, \label{r-rsph4} \\
        && \sum \limits_{j:(j,t) \in A}{y_{jt}} - \sum \limits_{k:(t,k) \in A}{y_{tk}} = 1, \label{r-rsph3} \\
        && \sum \limits_{(i,j) \in A}{d_{ij}y_{ij}} \leq \beta, \label{r-rsph5} \\
        && \mbox{Constraints (\ref{r-rspr1})-(\ref{r-rspr3})}, \nonumber \\
        && y_{ij} \in \{0,1\} \quad \forall\, (i,j) \in A.\label{r-rsph6}
\end{eqnarray}

The objective function in (\ref{r-rsph0}) represents the $\beta$-heuristic robustness cost of the path defined by the $y$ variables.  Constraints (\ref{r-rsph2})-(\ref{r-rsph5}) and (\ref{r-rsph6}) ensure that $y$ belongs to $\mathcal{P}(\beta)$.  Constraints (\ref{r-rspr1})-(\ref{r-rspr3}) are the remaining restrictions related to $\tilde{\mathcal{L}_1}$.

The heuristic consists in solving the corresponding problem $\mathcal{H}_1$ in order to find
a $\beta$-heuristic robust path $\tilde{p}^* \in \mathcal{P}(\beta)$.
Note that $\tilde{p}^*$ is also a feasible solution path for R-RSP, and, according to Proposition~\ref{prop01}, its $\beta$-heuristic
robustness cost provides an upper bound on the solution of R-RSP. Such bound can be improved by the evaluation of the actual $\beta$-restricted robustness
cost of $\tilde{p}^*$.
\subsection{The Robust Set Covering problem (RSC)}
The Robust Set Covering problem (RSC), introduced by Pereira and Averbakh \cite{Pereira11}, is an interval data min-max regret generalization of the Set Covering problem (SC).
The classical SC is known to be strongly NP-hard \cite{Garey79}, and, thus, the same holds for RSC.

Let $O = (o_{ij})$ be an $i \times j$ binary matrix, such that $I = \{1,\dots,i\}$ and $J = \{1,\dots,j\}$ are its
corresponding rows and columns sets, respectively.
We say that a column $j \in J$ \emph{covers} a row $i \in I$ if $o_{ij} = 1$. In this sense, a \emph{covering} is a subset $K \subseteq J$ of columns
such that every row in $I$ is covered by at least one column from $K$. Hereafter, we denote by $\Lambda$ the set of all possible coverings.

In the case of RSC, we associate with each column $j \in J$ a continuous cost interval [${l}_{j}$,${u}_{j}$], with $l_j,u_j \in \mathbb{Z}_+$ and ${l}_{j} \leq {u}_{j}$.
Accordingly, a scenario $s$ is an assignment of column costs, where a cost
$c_{j}^s \in [{l}_{j},{u}_{j}]$ is fixed for all $j \in J$.
The set of all these possible cost scenarios is denoted by $\mathcal{S}$, and the
cost of a covering $K \in \Lambda$ in a scenario $s \in \mathcal{S}$
is given by $C_K^s = \sum\limits_{j \in K}{c_{j}^s}$.
\begin{definition}
  A covering $K^*(s) \in \Lambda$ is said to be \emph{optimal}
  in a scenario $s \in \mathcal{S}$ if it has the smallest cost in $s$ among all coverings in
  $\Lambda$, i.e., $K^*(s) = \arg\min\limits_{K \in \Lambda}{C_{K}^s}$.
\end{definition}
\begin{definition}
  The \emph{regret (robust deviation)} of a covering $K \in \Lambda$ in
  a scenario $s \in \mathcal{S}$, denoted by $r^{s}_K$, is the
  difference between the cost $C^s_K$ of $K$ in $s$ and the cost of an
  optimal covering $K^*(s) \in \Lambda$ in $s$, i.e., $r^{s}_K = C^s_K -
  C^s_{K^*(s)}$.
\end{definition}
\begin{definition}
  The \emph{robustness cost} of a covering $K \in
  \Lambda$, denoted by $R_K$, is the maximum regret of $K$ among all possible scenarios, i.e., $R_K =
  \max\limits_{s \in \mathcal{S}}{r^{s}_K}$.
\end{definition}
\begin{definition}
  A covering $K^* \in \Lambda$ is said to be a
  \emph{robust covering} if it has the smallest
  robustness cost among all coverings in
  $\Lambda$, i.e., $K^* = \arg\min\limits_{K \in
    \Lambda}{R_{K}}$.
\end{definition}
\begin{definition}
  \emph{RSC} consists in finding a robust covering $K^* \in \Lambda$.
\end{definition}

As far as we are aware, RSC was only addressed in \cite{Pereira11} and \cite{Amadeu15}. On the other hand, the classical SC has been widely studied in the literature,
especially because it can be used as a basic model for applications in several fields, including production planning \cite{Cochran05} and crew management \cite{Caprara97}.
We refer to \cite{Ceria97} for an annotated bibliography on the main applications and the state-of-the-art of SC. 
We highlight that, as well as for the classical counterpart of R-RSP, directly solving compact ILP formulations for RSC via CPLEX's \emph{branch-and-bound} is competitive
with the best exact algorithms for the problem \cite{Caprara2000}. 
\subsubsection{Mathematical formulation}
The mathematical formulation presented below was proposed by Pereira and Averbakh \cite{Pereira11}. As the one for R-RSP presented in
Section~\ref{s_model}, this formulation is based on the modeling technique for interval min-max regret
problems in general discussed in Section~\ref{s_general_robust_pb}. For the specific case of RSC, Theorem~\ref{teo01} can be stated as follows.
\begin{proposition}
\label{prop_rsc_max}
Given a covering $X \in \Lambda$, the regret of $K$ is maximum in the
scenario induced by $K$, where the costs of all the columns of $K$ are in their corresponding upper bounds and the costs of all the other columns are in their
lower bounds.
 \end{proposition}
 
Now, consider the decision variables $y$ on the choice of columns belonging or not to a robust covering, such that $y_{j} = 1$ if the column $j \in J$
belongs to the solution; $y_{j} = 0$, otherwise. Also let $\rho$ be a continuous variable. The MILP formulation of \cite{Pereira11} is provided from \eqref{f_rsc01} to \eqref{f_rsc05}.
\begin{eqnarray}
  \min && \sum \limits_{j \in J}{u_{j}y_{j}} - \rho \label{f_rsc01} \\
  s.t. && \sum \limits_{j \in J}{o_{ij}y_j} \geq 1 \quad \forall\, i \in I, \label{f_rsc02} \\
        && \rho \leq \sum \limits_{j \in J}{( l_{j} + (u_{j} - l_{j})y_{j})\bar{x}_{j}} \quad \forall\, \bar{x} \in \Lambda, \label{f_rsc03} \\
        && y_{j} \in \{0,1\} \quad \forall\, j \in J, \label{f_rsc04} \\
        && \rho \mbox{ free}. \label{f_rsc05}
\end{eqnarray}

The objective function in (\ref{f_rsc01}) considers Proposition~\ref{prop_rsc_max} and gives the robustness cost of a robust covering.
Constraints (\ref{f_rsc02}) and (\ref{f_rsc04}) ensure that the $y$ variables represent a covering, whereas constraints (\ref{f_rsc03}) guarantee
that $\rho$ does not exceed the cost of an optimal covering in the scenario induced by the covering defined by $y$. In (\ref{f_rsc03}),
$\bar{x}$ is a constant vector, one for each possible covering. Constraint (\ref{f_rsc05}) gives the domain of the variable $\rho$.

As for R-RSP, we use the logic-based Benders' algorithm discussed in the Section~\ref{s_benders} to solve the formulation described above.
\subsubsection{An LP based heuristic for RSC}
Consider the following ILP formulation used to model the classical SC in a given cost scenario $s \in \mathcal{S}$. Here,
the binary variables $x$ define an optimal covering $K^*(s) \in \Lambda$ in $s$, such that $x_{j} = 1$ if the column
$j \in J$ belongs to $K^*(s)$, and $x_{j} = 0$, otherwise.
\begin{eqnarray}
  \mbox{$(\mathcal{I}_2)\quad$}\min && \sum \limits_{j \in J}{c^s_{j}x_{j}} \label{i_rsc01} \\
  s.t. && \sum \limits_{j \in J}{o_{ij}x_j} \geq 1 \quad \forall\, i \in I, \label{i_rsc02} \\
        && x_{j} \in \{0,1\} \quad \forall\, j \in J. \label{i_rsc03}
\end{eqnarray}

The objective function in (\ref{i_rsc01}) gives the cost, in the scenario $s$, of the covering defined by $x$,
while constraints (\ref{i_rsc02}) and (\ref{i_rsc03}) ensure that $x$ identifies a covering in $\Lambda$.
Relaxing the integrality on $x$, we obtain:
\begin{eqnarray}
  \mbox{$(\mathcal{L}_2)\quad$}\min && \sum \limits_{j \in J}{c^s_{j}x_{j}} \label{lp_rsc01} \\
  s.t. && \mbox{Constraint (\ref{i_rsc02}),} \nonumber \\
        && x_{j} \geq 0 \quad \forall\, j \in J. \label{lp_rsc02}
\end{eqnarray}

Notice that we omitted the domain constraints $x_{j} \leq 1$ for all $j \in J$, since they are redundant in this case.
Considering formulations $\mathcal{I}_2$, $\mathcal{L}_2$ and the generic model $\hat{\mathcal{R}}$, given by (\ref{r01}), (\ref{g05}), (\ref{g06}),
(\ref{g12}) and (\ref{g08}), we can devise a heuristic formulation for RSC. Now, the $y$ variables represent a heuristic solution covering, and
the variables $\{\lambda_i: i \in I\}$ are the ones of the dual problem related to $\mathcal{L}_2$.
\begin{eqnarray}
  \mbox{$(\mathcal{H}_2)\quad$}\min & & \Bigg(\sum \limits_{j \in J}{u_{j}y_{j}} - \sum \limits_{i \in I}{\lambda_{i}}\Bigg)\label{h_rsc01} \\
  s.t. 	&& \sum \limits_{j \in J}{o_{ij}y_j} \geq 1 \quad \forall\, i \in I, \label{h_rsc02} \\
	&& \sum \limits_{i \in I}{o_{ij}\lambda_i} \leq l_{j} + (u_{j}-l_{j})y_{j} \quad \forall\, j \in J, \label{h_rsc03} \\
	&& \lambda_{i} \geq 0 \quad \forall\, i \in I,\label{h_rsc04} \\
        && y_{j} \in \{0,1\} \quad \forall\, j \in J.\label{h_rsc05}
\end{eqnarray}

According to Proposition~\ref{prop01}, the objective function in (\ref{h_rsc01}) gives an upper bound on the robustness cost of a robust solution for RSC.
Constraints (\ref{h_rsc02}) and (\ref{h_rsc05}) ensure that $y$ defines a covering in $\Lambda$, while constraints (\ref{h_rsc03}) are the ones
related to the dual of ${\mathcal{L}_2}$.
Notice that constraints (\ref{h_rsc03}) consider the cost of each column $j \in J$ in the scenario induced by the covering identified by the $y$ variables.
Restrictions (\ref{h_rsc04}) give the domain of the dual variables $\lambda$.

In the case of RSC, the LP based heuristic consists in solving formulation $\mathcal{H}_2$ and, then, computing the robustness cost of the solution obtained.
\section{Computational experiments}
\label{s_results}
In this section, we evaluate, out of computational experiments, the effectiveness and the time efficiency of the proposed heuristic at solving
the two interval robust-hard problems considered in this study.
For short, the Logic-Based Benders' decomposition algorithm is referred to as LB-Benders', whereas the LP based Heuristic is referred to as LPH.
LB-Benders', LPH and the 2-approximation heuristic for interval min-max regret problems, namely AMU~\cite{Kasperski05}, were implemented in C++,
along with the optimization solver ILOG CPLEX 12.5.  The computational experiments were performed on a 64 bits
Intel$^{\textregistered}$ Xeon$^{\textregistered}$ E5405 machine with 2.0 GHz and 7.0 GB of RAM, under Linux operating system. LB-Benders' was set to run
for up to 3600 seconds of wall-clock time.
\subsection{The Restricted Robust Shortest Path problem (R-RSP)}
In all of the algorithms implemented, whenever a classical R-SP instance had to be solved, we used CPLEX to handle the ILP formulation $\mathcal{I}_1$, defined by (\ref{r-sp0})-(\ref{r-sp5}),
directly. We also used CPLEX to solve each master problem in LB-Benders' and the heuristic formulation $\mathcal{H}_1$, defined by (\ref{r-rsph0})-(\ref{r-rsph6})
\subsubsection{Benchmarks description}
Due to the lack of R-RSP instances in the literature, we generated two benchmarks of instances inspired by the applications described in Section~\ref{s_description}. These benchmarks were adapted from two sets of RSP instances: \emph{Kara\c{s}an} \cite{Karasan01} and \emph{Coco} \cite{amadeu14} instances, which model, respectively, telecommunications and urban transportation networks.

Kara\c{s}an instances have been largely used in experiments concerning RSP \cite{MonteGa04c,Karasan01,MonteGa04b,amadeu14,MonteGa04a}. They consist of layered \cite{Sug81} and acyclic \cite{Bondy76} digraphs. In these digraphs, each of the $\kappa$ layers has the same number $\omega$ of vertices. There is an arc from every vertex in a layer $b \in \{1,\dots,\kappa-1\}$ to every vertex in the adjacent layer $b+1$. Moreover, there is an arc from the origin $o$ to every vertex in the first layer, and an arc from every vertex in the layer $\kappa$ to the destination vertex $t$.
These instances are named K-$v$-${\Phi}_{max}$-$\delta$-$\omega$, where $v$ is the number of vertices (aside from $o$ and $t$),  
${\Phi}_{max}$ is an integer constant, and $0 < \delta < 1$ is a continuous value. The arc cost intervals were generated as follows.
For each arc $(i,j) \in A$, a random integer value ${\Phi}_{ij}$ was uniformly chosen in the range $[1,{\Phi}_{max}]$.
Afterwards, random integer values $l_{ij}$ and $u_{ij}$ were uniformly selected, respectively, in the ranges
$[(1-\delta)\cdot{\Phi}_{ij},(1+\delta)\cdot{\Phi}_{ij}]$ and $[l_{ij},(1+\delta)\cdot{\Phi}_{ij}]$.  Note that ${\Phi}$ plays the role of a base-case
scenario, and $\delta$ determines the degree of uncertainty. Figure~\ref{fig_inst_a} shows an example of an acyclic digraph with 3 layers of width 2.
\begin{figure}[!ht]
\center
\includegraphics*[scale=0.81]{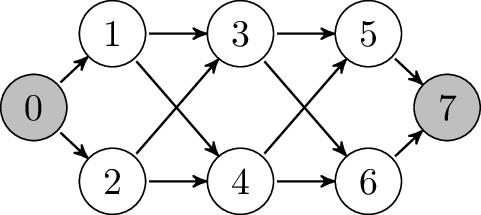}
\caption{An acyclic digraph with 3 layers of width 2. Here, $o = 0$ and $t = 7$.}
\label{fig_inst_a}
\end{figure}

Coco instances consist of grid digraphs based on $n \times m$ matrices, where $n$ is the number of rows and $m$ is the number of columns. Each matrix cell corresponds to a vertex in the digraph, and there are two bidirectional arcs between each pair of vertices whose respective matrix cells are adjacent. The origin $o$ is defined as the upper left vertex, and the destination $t$ is defined as the lower right vertex. These instances are named G-$n$$\times$$m$-${\Phi}_{max}$-$\delta$, with $0 < \delta < 1$, where ${\Phi}_{max}$ is an integer value. Given ${\Phi}_{max}$ and $\delta$ values, the cost intervals are generated as in the Kara\c{s}an instances.  Figure~\ref{fig_inst_b} gives an example of a grid digraph.
\begin{figure}[!ht]
\center
\includegraphics*[scale=0.87]{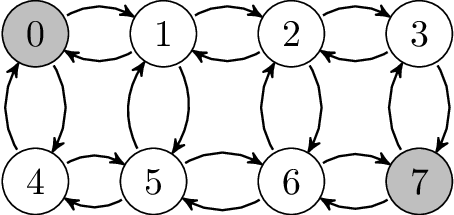}
\caption{A $2 \times 4$ grid digraph, with $o = 0$ and $t = 7$.}
\label{fig_inst_b}
\end{figure}

For all instances, the resource consumption associated with each arc is given by a random integer value uniformly selected in the interval $(0,10]$. The small interval amplitude allows the generation of instances in which most of the arcs are candidate to appear in an optimal solution, increasing the number of feasible solutions. The symmetry with respect to arc resource consumptions was preserved, i.e., we considered $d_{ij} = d_{ji}$ for any pair of adjacent vertices $i$ and $j$ such that $(i,j) \in A$ and $(j,i) \in A$.

The resource consumption limit $\beta$ of a given instance was computed as follows. Consider the set $\mathcal{P}$ of all the paths from $o$ to $t$, and let $\bar{p} \in \mathcal{P}$ be a shortest path in terms of resource consumption, i.e., $\bar{p} = \arg\min\limits_{p \in \mathcal{P}}{D_{p}}$. We set $\beta = 1.1 \cdot D_{\bar{p}}$, which means that is given a $10\%$ tolerance with respect to the minimum resource consumption $D_{\bar{p}}$. This way, the resource consumption limit is tighter.

We generated Kara\c{s}an and Coco instances of 1000 and 2000 vertices, with
$\Phi_{max} \in \{20, 200\}$, $\delta \in \{0.5, 0.9\}$ and $\omega \in \{5, 10, 25\}$. Considering these values, a group of 10 instances was generated for each possible parameters
configuration. In summary, 480 instances were used in the experiments.
\subsubsection{Results}
Computational experiments were carried out in order to evaluate if the proposed heuristic efficiently finds optimal or near-optimal solutions for the two benchmarks of instances described above.
Results for Kara\c{s}an and Coco instances are reported in Tables \ref{table01} and \ref{table02}, respectively. The first column displays the name of each group of 10 instances. The second and third columns show, respectively, the number of instances solved at optimality by LB-Benders' within 3600 seconds, and the average wall-clock processing time (in seconds) spent in solving these instances. If no instance in the group was solved at optimality, this entry is filled with a dash.
The fourth and fifth columns show, respectively, the average and the standard deviation (over the 10 instances) of the relative optimality gaps given by $100
\cdot \frac{UB_b - LB_b}{UB_b}$. Here, $LB_b$ and $UB_b$ are, respectively, the best lower and upper bounds obtained by LB-Benders' for a given instance.
The sixth column displays the average wall-clock processing time (in seconds) of AMU. The seventh column shows the average (over the 10 instances) of the relative gaps given by $100 \cdot \frac{UB_{amu} - LB_b}{UB_{amu}}$,
where $UB_{amu}$ is the best upper bound obtained by AMU for a given instance. The standard deviation of these gaps is given in the eighth column.
Likewise, the ninth column shows the average wall-clock processing time of LPH, and the last two columns give the average and
the standard deviation (over the 10 instances) of the gaps given by $100 \cdot \frac{UB_{lph} - LB_b}{UB_{lph}}$.
Here, $UB_{lph}$ is the $\beta$-restricted robustness cost of the solution obtained by LPH for a given instance.

Notice that the gaps referred to the solutions obtained by AMU and LPH consider the best lower bounds obtained by LB-Benders' (within 3600 seconds of
execution), which might not correspond to the cost of optimal solutions. Thus, the aforementioned gaps may
overestimate the actual gaps between the cost of the solutions obtained and the cost of optimal ones.

\begin{landscape}
\begin{table}[!h]
\caption{Computational results for the layered and acyclic digraph instances.}
\label{table01}
\center
\scalebox{0.7}{
\begin{tabular}{lrrrrrrrrrrrr}
\toprule
\multicolumn{1}{r}{} & \multicolumn{4}{c}{\textbf{LB-Benders'}} & & \multicolumn{3}{c}{\textbf{AMU}} & & \multicolumn{3}{c}{\textbf{LPH}}\\
\cmidrule{2-5} \cmidrule{7-9} \cmidrule{11-13}
Test set & {\#opt} & {Time (s)} & {AvgGAP (\%)} & {StDev (\%)} & & Time (s) & {AvgGAP (\%)} & {StDev (\%)} & & Time (s) & {AvgGAP (\%)} & {StDev (\%)}\\
\midrule
K-1000-20-0.5-5 & 8~~ & 1515.25~~ &0.11~~~~~~ & 0.24~~~~ & & 3.13~~ & 4.21~~~~~~  & 3.59~~~~ & &10.33~~ &0.11~~~~~~  &0.24~~~~~\\
K-1000-20-0.9-5 & 2~~ & 2108.61~~ &7.47~~~~~~ & 5.29~~~~ & & 3.12~~ & 10.06~~~~~~  & 4.21~~~~ & &28.84~~ &5.79~~~~~~  &3.75~~~~~\\
K-1000-200-0.5-5 & 10~~ & 1165.08~~ &0.00~~~~~~ & 0.00~~~~ & & 3.13~~ & 3.68~~~~~~  & 2.18~~~~ & &8.52~~ &0.16~~~~~~  &0.28~~~~~\\
K-1000-200-0.9-5 & 1~~ & 1919.16~~ &4.85~~~~~~ & 3.30~~~~ & & 3.07~~ & 9.77~~~~~~  & 1.83~~~~ & &21.36~~ &3.80~~~~~~  &2.57~~~~~\\
\cmidrule(lr){1-13}
K-1000-20-0.5-10 & 10~~ & 83.30~~ &0.00~~~~~~ & 0.00~~~~ & & 4.15~~ & 3.06~~~~~~  & 2.70~~~~ & &11.23~~ &0.07~~~~~~  &0.22~~~~~\\
K-1000-20-0.9-10 & 10~~ & 218.69~~ &0.00~~~~~~ & 0.00~~~~ & & 4.17~~ & 4.52~~~~~~  & 3.20~~~~ & &17.75~~ &0.00~~~~~~  &0.00~~~~~\\
K-1000-200-0.5-10 & 10~~ & 42.54~~ &0.00~~~~~~ & 0.00~~~~ & & 3.96~~ & 1.57~~~~~~  & 1.84~~~~ & &10.45~~ &0.31~~~~~~  &0.97~~~~~\\
K-1000-200-0.9-10 & 10~~ & 413.43~~ &0.00~~~~~~ & 0.00~~~~ & & 3.94~~ & 4.50~~~~~~  & 2.98~~~~ & &24.44~~ &0.05~~~~~~  &0.15~~~~~\\
\cmidrule(lr){1-13}
K-1000-20-0.5-25 & 10~~ & 17.61~~ &0.00~~~~~~ & 0.00~~~~ & & 7.35~~ & 3.12~~~~~~  & 4.52~~~~ & &21.55~~ &0.00~~~~~~  &0.00~~~~~\\
K-1000-20-0.9-25 & 10~~ & 32.22~~ &0.00~~~~~~ & 0.00~~~~ & & 7.51~~ & 2.35~~~~~~  & 3.69~~~~ & &26.97~~ &0.24~~~~~~  &0.76~~~~~\\
K-1000-200-0.5-25 & 10~~ & 18.17~~ &0.00~~~~~~ & 0.00~~~~ & & 7.25~~ & 1.27~~~~~~  & 2.71~~~~ & &23.67~~ &0.00~~~~~~  &0.00~~~~~\\
K-1000-200-0.9-25 & 10~~ & 41.04~~ &0.00~~~~~~ & 0.00~~~~ & & 7.33~~ & 1.13~~~~~~  & 1.89~~~~ & &30.92~~ &0.00~~~~~~  &0.00~~~~~\\
\cmidrule(lr){1-13}
K-2000-20-0.5-5 & 0~~ & 0.00~~ &14.47~~~~~~ & 5.89~~~~ & & 10.61~~ & 14.30~~~~~~  & 4.06~~~~ & &64.61~~ &10.43~~~~~~  &3.87~~~~~\\
K-2000-20-0.9-5 & 0~~ & 0.00~~ &25.01~~~~~~ & 2.69~~~~ & & 10.68~~ & 21.35~~~~~~  & 2.17~~~~ & &241.13~~ &17.94~~~~~~  &2.38~~~~~\\
K-2000-200-0.5-5 & 0~~ & 0.00~~ &14.45~~~~~~ & 3.10~~~~ & & 10.38~~ & 14.15~~~~~~  & 2.87~~~~ & &69.62~~ &10.70~~~~~~  &2.21~~~~~\\
K-2000-200-0.9-5 & 0~~ & 0.00~~ &25.77~~~~~~ & 3.39~~~~ & & 10.43~~ & 22.13~~~~~~  & 3.17~~~~ & &454.79~~ &18.15~~~~~~  &2.79~~~~~\\
\cmidrule(lr){1-13}
K-2000-20-0.5-10 & 8~~ & 1297.89~~ &0.91~~~~~~ & 2.23~~~~ & & 13.01~~ & 4.37~~~~~~  & 3.69~~~~ & &92.04~~ &0.89~~~~~~  &1.99~~~~~\\
K-2000-20-0.9-10 & 0~~ & 0.00~~ &7.34~~~~~~ & 4.06~~~~ & & 12.75~~ & 9.46~~~~~~  & 4.30~~~~ & &252.91~~ &5.91~~~~~~  &3.45~~~~~\\
K-2000-200-0.5-10 & 4~~ & 846.44~~ &1.37~~~~~~ & 2.30~~~~ & & 12.34~~ & 4.12~~~~~~  & 2.94~~~~ & &112.94~~ &1.04~~~~~~  &1.66~~~~~\\
K-2000-200-0.9-10 & 0~~ & 0.00~~ &5.99~~~~~~ & 2.44~~~~ & & 12.39~~ & 8.78~~~~~~  & 3.54~~~~ & &211.80~~ &4.77~~~~~~  &2.31~~~~~\\
\cmidrule(lr){1-13}
K-2000-20-0.5-25 & 10~~ & 155.30~~ &0.00~~~~~~ & 0.00~~~~ & & 21.72~~ & 5.98~~~~~~  & 7.21~~~~ & &122.22~~ &0.00~~~~~~  &0.00~~~~~\\
K-2000-20-0.9-25 & 10~~ & 408.79~~ &0.00~~~~~~ & 0.00~~~~ & & 21.60~~ & 1.46~~~~~~  & 1.54~~~~ & &222.00~~ &0.07~~~~~~  &0.22~~~~~\\
K-2000-200-0.5-25 & 10~~ & 138.88~~ &0.00~~~~~~ & 0.00~~~~ & & 21.88~~ & 1.53~~~~~~  & 2.11~~~~ & &121.82~~ &0.04~~~~~~  &0.13~~~~~\\
K-2000-200-0.9-25 & 10~~ & 572.74~~ &0.00~~~~~~ & 0.00~~~~ & & 21.18~~ & 2.70~~~~~~  & 2.76~~~~ & &215.21~~ &0.02~~~~~~  &0.05~~~~~\\
\cmidrule{1-13}
\textbf{Average}        & & & 4.49~~~~~~ & 1.46~~~~ & & & 6.65~~~~~~ & 3.15~~~~~ & & &3.35~~~~~~ &1.25~~~~~ \\
\bottomrule
\end{tabular}}
\end{table}
\end{landscape}

\begin{landscape}
\begin{table}[!h]
\caption{Computational results for the grid digraph instances.}
\label{table02}
\center
\scalebox{0.7}{
\begin{tabular}{lrrrrrrrrrrrr}
\toprule
\multicolumn{1}{r}{} & \multicolumn{4}{c}{\textbf{LB-Benders'}} & & \multicolumn{3}{c}{\textbf{AMU}} & & \multicolumn{3}{c}{\textbf{LPH}}\\
\cmidrule{2-5} \cmidrule{7-9} \cmidrule{11-13}
Test set & {\#opt} & {Time (s)} & {AvgGAP (\%)} & {StDev (\%)} & & Time (s) & {AvgGAP (\%)} & {StDev (\%)} & & Time (s) & {AvgGAP (\%)} & {StDev (\%)}\\
\midrule
G-32x32-20-0.5 & 10~~ & 7.31~~ &0.00~~~~~~ & 0.00~~~~ & & 3.21~~ & 4.16~~~~~~  & 9.20~~~~ & &3.80~~ &0.22~~~~~~  &0.70~~~~~\\
G-32x32-20-0.9 & 10~~ & 8.63~~ &0.00~~~~~~ & 0.00~~~~ & & 3.43~~ & 5.38~~~~~~  & 7.49~~~~ & &5.01~~ &0.69~~~~~~  &1.67~~~~~\\
G-32x32-200-0.5 & 10~~ & 6.83~~ &0.00~~~~~~ & 0.00~~~~ & & 3.35~~ & 1.56~~~~~~  & 2.72~~~~ & &3.88~~ &3.11~~~~~~  &4.72~~~~~\\
G-32x32-200-0.9 & 10~~ & 9.79~~ &0.00~~~~~~ & 0.00~~~~ & & 3.22~~ & 3.90~~~~~~  & 5.89~~~~ & &5.33~~ &0.00~~~~~~  &0.00~~~~~\\
\cmidrule(lr){1-13}
G-20x50-20-0.5 & 10~~ & 6.66~~ &0.00~~~~~~ & 0.00~~~~ & & 2.86~~ & 2.71~~~~~~  & 5.94~~~~ & &3.96~~ &0.42~~~~~~  &1.34~~~~~\\
G-20x50-20-0.9 & 10~~ & 9.14~~ &0.00~~~~~~ & 0.00~~~~ & & 3.00~~ & 1.80~~~~~~  & 2.93~~~~ & &4.76~~ &0.76~~~~~~  &1.10~~~~~\\
G-20x50-200-0.5 & 10~~ & 6.23~~ &0.00~~~~~~ & 0.00~~~~ & & 3.24~~ & 2.54~~~~~~  & 5.09~~~~ & &4.47~~ &0.78~~~~~~  &2.45~~~~~\\
G-20x50-200-0.9 & 10~~ & 13.43~~ &0.00~~~~~~ & 0.00~~~~ & & 3.15~~ & 2.33~~~~~~  & 4.43~~~~ & &5.43~~ &0.58~~~~~~  &0.74~~~~~\\
\cmidrule(lr){1-13}
G-5x200-20-0.5 & 10~~ & 397.06~~ &0.00~~~~~~ & 0.00~~~~ & & 2.92~~ & 5.29~~~~~~  & 4.89~~~~ & &10.97~~ &0.10~~~~~~  &0.31~~~~~\\
G-5x200-20-0.9 & 9~~ & 670.84~~ &0.28~~~~~~ & 0.89~~~~ & & 2.91~~ & 4.64~~~~~~  & 1.69~~~~ & &19.83~~ &0.36~~~~~~  &0.87~~~~~\\
G-5x200-200-0.5 & 10~~ & 153.69~~ &0.00~~~~~~ & 0.00~~~~ & & 2.97~~ & 4.82~~~~~~  & 3.57~~~~ & &8.75~~ &0.29~~~~~~  &0.66~~~~~\\
G-5x200-200-0.9 & 8~~ & 1059.09~~ &0.20~~~~~~ & 0.50~~~~ & & 3.01~~ & 8.00~~~~~~  & 3.73~~~~ & &22.34~~ &0.32~~~~~~  &0.57~~~~~\\
\cmidrule(lr){1-13}
G-44x44-20-0.5 & 10~~ & 23.52~~ &0.00~~~~~~ & 0.00~~~~ & & 10.36~~ & 0.97~~~~~~  & 1.92~~~~ & &13.15~~ &0.81~~~~~~  &1.42~~~~~\\
G-44x44-20-0.9 & 10~~ & 29.43~~ &0.00~~~~~~ & 0.00~~~~ & & 10.32~~ & 3.92~~~~~~  & 8.06~~~~ & &15.50~~ &0.55~~~~~~  &1.04~~~~~\\
G-44x44-200-0.5 & 10~~ & 23.00~~ &0.00~~~~~~ & 0.00~~~~ & & 9.77~~ & 5.05~~~~~~  & 4.02~~~~ & &14.58~~ &0.00~~~~~~  &0.00~~~~~\\
G-44x44-200-0.9 & 10~~ & 38.13~~ &0.00~~~~~~ & 0.00~~~~ & & 9.53~~ & 2.70~~~~~~  & 3.23~~~~ & &19.98~~ &0.13~~~~~~  &0.22~~~~~\\
\cmidrule(lr){1-13}
G-20x100-20-0.5 & 10~~ & 34.77~~ &0.00~~~~~~ & 0.00~~~~ & & 10.90~~ & 2.81~~~~~~  & 5.26~~~~ & &20.38~~ &0.00~~~~~~  &0.00~~~~~\\
G-20x100-20-0.9 & 10~~ & 90.28~~ &0.00~~~~~~ & 0.00~~~~ & & 10.76~~ & 5.41~~~~~~  & 4.98~~~~ & &29.39~~ &0.61~~~~~~  &1.13~~~~~\\
G-20x100-200-0.5 & 10~~ & 47.27~~ &0.00~~~~~~ & 0.00~~~~ & & 11.39~~ & 1.64~~~~~~  & 1.97~~~~ & &20.17~~ &0.25~~~~~~  &0.66~~~~~\\
G-20x100-200-0.9 & 10~~ & 59.16~~ &0.00~~~~~~ & 0.00~~~~ & & 10.53~~ & 3.37~~~~~~  & 3.27~~~~ & &26.32~~ &0.21~~~~~~  &0.37~~~~~\\
\cmidrule(lr){1-13}
G-5x400-20-0.5 & 1~~ & 3444.71~~ &2.57~~~~~~ & 1.98~~~~ & & 10.26~~ & 6.82~~~~~~  & 3.90~~~~ & &62.50~~ &2.14~~~~~~  &1.62~~~~~\\
G-5x400-20-0.9 & 0~~ & 0.00~~ &7.76~~~~~~ & 3.95~~~~ & & 10.20~~ & 10.83~~~~~~  & 3.16~~~~ & &150.45~~ &5.67~~~~~~  &2.89~~~~~\\
G-5x400-200-0.5 & 1~~ & 2719.17~~ &5.24~~~~~~ & 4.32~~~~ & & 10.27~~ & 8.67~~~~~~  & 3.59~~~~ & &79.37~~ &3.94~~~~~~  &3.17~~~~~\\
G-5x400-200-0.9 & 0~~ & 0.00~~ &11.50~~~~~~ & 3.45~~~~ & & 9.93~~ & 13.27~~~~~~  & 3.75~~~~ & &249.81~~ &8.19~~~~~~  &2.24~~~~~\\
\cmidrule{1-13}
\textbf{Average}       & & & 1.15~~~~~~ & 0.63~~~~ & & & 4.69~~~~~~ & 4.36~~~~~ & & &1.26~~~~~~ &1.24~~~~~ \\
\bottomrule
\end{tabular}}
\end{table}
\end{landscape}

With respect to Kara\c{s}an instances (Table~\ref{table01}), the average gaps referred to the solutions provided by LB-Benders', AMU and
LPH are up to, respectively, 7.47\%, 10.06\% and 5.79\% for the instances with 1000 vertices (see K-1000-20-0.9-5).
For the instances with 2000 vertices, the average gaps referred to the solutions provided by LB-Benders', AMU and
LPH are up to, respectively, 25.77\%, 22.13\% and 18.15\% (see K-2000-200-0.9-5).
Notice that the average gaps of AMU over all Kara\c{s}an instances is 6.65\%, while that of LPH is only 3.35\%.
In fact, the average gaps of the solutions provided by LPH are smaller than those of AMU for all the sets of instances considered.
It can also be observed that the smaller the value of $\omega$ is, the larger the average gaps achieved by the three algorithms are.
Furthermore, for the hardest instances (with $\omega = 5$), the average gaps of the solutions provided by LPH are smaller than or equal to those of LB-Benders'
(except for K-1000-200-0.5-5).

Regarding Coco instances (Table~\ref{table02}), it can be seen that the average optimality gaps referred to the solutions provided by
LB-Benders', AMU and LPH are at most, respectively, 11.50\%, 13.27\% and 8.19\% (see G-5x400-200-0.9).
Also in this case, the average gaps of the solutions provided by LPH are smaller than those of AMU for all the sets of instances (except for G-32x32-200-0.5).
The average optimality gap of LB-Benders' over all Coco instances is very small (1.15\%), while that of LPH is very close to this value (1.26\%).
It can also be observed that the instances based on 5x400 grids are much harder to solve than the other grid instances.
Moreover, the average relative gaps referred to the solutions provided by LPH is always smaller than those of LB-Benders'
for the hardest grid instances.

For both benchmarks, LPH clearly outperforms AMU in terms of the quality of the solutions obtained. In fact, the heuristic is even able to achieve better
upper bounds than LB-Benders' at solving some instance sets, especially the hardest ones.
Furthermore, LPH required significantly smaller computational effort than LB-Benders'.
Therefore, LPH arises as an effective and time efficient heuristic to solve Kara\c{s}an and Coco instances.
One may note that, although the bounds provided by AMU are not as good as those given by LPH, the 2-approximation procedure performs within
tiny average computational times when compared with LB-Benders', especially for the hardest instances.
Thus, the results suggest that AMU should be considered in the cases where time efficiency is preferable to the quality of the solutions.

Our computational results also suggest that, for both benchmarks, instances generated with a higher degree of uncertainty (in particular, $\delta = 0.9$)
tend to become more difficult to be handled by LB-Benders' and LPH. As pointed out in related works \cite{Averbakh13,Karasan01}, high $\delta$ values
can increase the occurrence of overlapping cost intervals and, thus, decrease the number of dominated arcs. As a consequence,
the task of finding a robust path becomes more difficult. Such behaviour does not apply to AMU, as it simply solves R-SP instances in specific scenarios
that do not rely on the amplitude of the cost intervals.

In Tables~\ref{table03} and \ref{table04}, we detail the improvement of LPH over AMU in terms of the bounds achieved for Kara\c{s}an and Coco instances,
respectively. The first column displays the name of each instance set, and the second one shows the number of times (out of 10) that LPH gives better bounds
than AMU. The third and fourth columns show, respectively, the minimum and the maximum percentage improvements
given by $100 \cdot \frac{UB_{amu} - UB_{lph}}{UB_{amu}}$, while the last two columns give the average and the standard deviation (over the 10 instances
in each set) of these values. Notice that the percentage improvement assumes a negative value whenever AMU gives a better bound than LPH.

Regarding Kara\c{s}an instances (Table~\ref{table03}), the percentage improvement is up to 21.47\% (see K-2000-20-0.5-25). In fact, AMU only overcomes
LPH on two instances from the whole benchmark, one from K-1000-20-0.9-25 and other from K-2000-200-0.9-25.
For the grid digraph instances (Table~\ref{table04}), the percentage improvement is up to 29.31\% (see G-32x32-20-0.5). Moreover, AMU gives better bounds
than LPH only on 13 out of the 240 instances tested. We could not identify a precise pattern on these instances that might indicate why AMU performs better. However, most of them belong to
instance sets based on square-like grids (such as 32x32 and 44x44). Thus, we conjecture that the balance between the number of lines and columns in the
corresponding grids plays a role in such behaviour.
\begin{landscape}
\begin{table}[!h]
\caption{Details on the improvement of LPH over AMU for the two benchmarks of R-RSP instances.}
\center
\subfloat[Layered and acyclic digraph instances.]
{
\label{table03}
\scalebox{0.7}{
\begin{tabular}{lrrrrr}
\toprule
\multicolumn{2}{r}{} & \multicolumn{4}{c}{\textbf{LPH Improvement}}\\
\cmidrule{3-6}
Test set & {\#Better} & {Min (\%)} & {Max (\%)} & {Avg (\%)} & {StDev (\%)}\\
\midrule
K-1000-20-0.5-5 & 10&0.00~~~~&9.79~~~~~~&4.10~~~~~&3.59~~~~~~~\\
K-1000-20-0.9-5 & 10&0.60~~~~&9.62~~~~~~&4.51~~~~~&3.22~~~~~~~\\
K-1000-200-0.5-5 & 10&0.68~~~~&8.15~~~~~~&3.52~~~~~&2.26~~~~~~~\\
K-1000-200-0.9-5 & 10&2.06~~~~&12.49~~~~~~&6.14~~~~~&3.28~~~~~~~\\
\midrule
K-1000-20-0.5-10 & 10&0.00~~~~&7.79~~~~~~&2.99~~~~~&2.56~~~~~~~\\
K-1000-20-0.9-10 & 10&0.00~~~~&12.10~~~~~~&4.52~~~~~&3.20~~~~~~~\\
K-1000-200-0.5-10 & 10&0.00~~~~&4.18~~~~~~&1.27~~~~~&1.51~~~~~~~\\
K-1000-200-0.9-10 & 10&1.33~~~~&12.16~~~~~~&4.45~~~~~&2.97~~~~~~~\\
\midrule
K-1000-20-0.5-25 & 10&0.00~~~~&11.90~~~~~~&3.12~~~~~&4.52~~~~~~~\\
K-1000-20-0.9-25 & 9&-2.47~~~~&10.39~~~~~~&2.10~~~~~&3.94~~~~~~~\\
K-1000-200-0.5-25 & 10&0.00~~~~&7.27~~~~~~&1.27~~~~~&2.71~~~~~~~\\
K-1000-200-0.9-25 & 10&0.00~~~~&5.68~~~~~~&1.13~~~~~&1.89~~~~~~~\\
\midrule
K-2000-20-0.5-5 & 10&0.41~~~~&9.88~~~~~~&4.29~~~~~&3.20~~~~~~~\\
K-2000-20-0.9-5 & 10&1.34~~~~&8.01~~~~~~&4.13~~~~~&2.09~~~~~~~\\
K-2000-200-0.5-5 & 10&0.39~~~~&7.15~~~~~~&3.86~~~~~&2.21~~~~~~~\\
K-2000-200-0.9-5 & 10&1.61~~~~&8.25~~~~~~&4.86~~~~~&2.22~~~~~~~\\
\midrule
K-2000-20-0.5-10 & 10&0.00~~~~&9.80~~~~~~&3.52~~~~~&2.81~~~~~~~\\
K-2000-20-0.9-10 & 10&1.48~~~~&6.72~~~~~~&3.80~~~~~&1.79~~~~~~~\\
K-2000-200-0.5-10 & 10&0.00~~~~&5.66~~~~~~&3.12~~~~~&1.95~~~~~~~\\
K-2000-200-0.9-10 & 10&1.04~~~~&8.19~~~~~~&4.23~~~~~&2.31~~~~~~~\\
\midrule
K-2000-20-0.5-25 & 10&0.00~~~~&21.57~~~~~~&5.98~~~~~&7.21~~~~~~~\\
K-2000-20-0.9-25 & 10&0.00~~~~&3.91~~~~~~&1.39~~~~~&1.60~~~~~~~\\
K-2000-200-0.5-25 & 10&0.00~~~~&6.46~~~~~~&1.49~~~~~&2.14~~~~~~~\\
K-2000-200-0.9-25 & 9&-0.14~~~~&7.20~~~~~~&2.68~~~~~&2.78~~~~~~~\\
\bottomrule
\end{tabular}
}
}
\subfloat[Grid digraph instances.]
{
\label{table04}
\scalebox{0.7}{
\begin{tabular}{lrrrrr}
\toprule
\multicolumn{2}{r}{} & \multicolumn{4}{c}{\textbf{LPH Improvement}}\\
\cmidrule{3-6}
Test set & {\#Better} & {Min (\%)} & {Max (\%)} & {Avg (\%)} & {StDev (\%)}\\
\midrule
G-32x32-20-0.5 & 10&0.00~~~~&29.31~~~~~~&3.94~~~~~&9.17~~~~~~~\\
G-32x32-20-0.9 & 9&-5.49~~~~&20.17~~~~~~&4.69~~~~~&7.95~~~~~~~\\
G-32x32-200-0.5 & 7&-11.21~~~~&0.00~~~~~~&-1.73~~~~~&3.71~~~~~~~\\
G-32x32-200-0.9 & 10&0.00~~~~&18.21~~~~~~&3.90~~~~~&5.89~~~~~~~\\
\midrule
G-20x50-20-0.5 & 10&0.00~~~~&18.89~~~~~~&2.29~~~~~&5.97~~~~~~~\\
G-20x50-20-0.9 & 9&-2.40~~~~&8.89~~~~~~&1.03~~~~~&3.14~~~~~~~\\
G-20x50-200-0.5 & 9&-8.41~~~~&15.57~~~~~~&1.70~~~~~&6.14~~~~~~~\\
G-20x50-200-0.9 & 9&-1.04~~~~&12.78~~~~~~&1.77~~~~~&4.13~~~~~~~\\
\midrule
G-5x200-20-0.5 & 10&0.00~~~~&14.34~~~~~~&5.20~~~~~&4.89~~~~~~~\\
G-5x200-20-0.9 & 10&2.58~~~~&8.01~~~~~~&4.29~~~~~&1.75~~~~~~~\\
G-5x200-200-0.5 & 10&0.00~~~~&11.14~~~~~~&4.54~~~~~&3.51~~~~~~~\\
G-5x200-200-0.9 & 10&1.22~~~~&12.15~~~~~~&7.70~~~~~&3.82~~~~~~~\\
\midrule
G-44x44-20-0.5 & 8&-3.95~~~~&5.77~~~~~~&0.15~~~~~&2.40~~~~~~~\\
G-44x44-20-0.9 & 9&-3.29~~~~&23.26~~~~~~&3.40~~~~~&7.95~~~~~~~\\
G-44x44-200-0.5 & 10&0.00~~~~&11.05~~~~~~&5.05~~~~~&4.02~~~~~~~\\
G-44x44-200-0.9 & 9&-0.47~~~~&7.45~~~~~~&2.56~~~~~&3.32~~~~~~~\\
\midrule
G-20x100-20-0.5 & 10&0.00~~~~&14.09~~~~~~&2.81~~~~~&5.26~~~~~~~\\
G-20x100-20-0.9 & 10&0.00~~~~&11.06~~~~~~&4.85~~~~~&4.25~~~~~~~\\
G-20x100-200-0.5 & 10&0.00~~~~&5.80~~~~~~&1.40~~~~~&1.80~~~~~~~\\
G-20x100-200-0.9 & 8&-0.66~~~~&7.65~~~~~~&3.17~~~~~&3.38~~~~~~~\\
\midrule
G-5x400-20-0.5 & 10&1.27~~~~&10.58~~~~~~&4.80~~~~~&2.87~~~~~~~\\
G-5x400-20-0.9 & 10&1.59~~~~&13.90~~~~~~&5.42~~~~~&3.65~~~~~~~\\
G-5x400-200-0.5 & 10&0.16~~~~&10.92~~~~~~&4.88~~~~~&3.63~~~~~~~\\
G-5x400-200-0.9 & 10&1.72~~~~&11.72~~~~~~&5.53~~~~~&3.26~~~~~~~\\
\bottomrule
\end{tabular}
}
}
\end{table}
\end{landscape}
\subsection{The Robust Set Covering problem (RSC)}
In all of the algorithms regarding RSC, we used CPLEX to handle the ILP formulation $\mathcal{I}_2$, defined by (\ref{i_rsc01})-(\ref{i_rsc03}), whenever a
classical SC instance had to be solved. We also used CPLEX to solve each master problem in LB-Benders' and the heuristic formulation $\mathcal{H}_2$,
defined by (\ref{h_rsc01})-(\ref{h_rsc05}).

\subsubsection{Benchmarks description}
In our experiments, we considered three benchmarks of instances from the literature of RSC, namely \emph{Beasley}, \emph{Montemanni} and
\emph{Kasperski-Zieli\'{n}ski} benchmarks \cite{Pereira11}. The three of them are based on classical SC instances from the OR-Library \cite{Beasley90}. However,
the way the column cost intervals are generated differs from benchmark to benchmark.

Regarding Beasley instances, let $\Phi_j$ represent the cost of a
column $j \in J$ in the original SC instance, and let $0 < \delta <1$ be a continuous value used to control the level of uncertainty referred to an RSC instance. For each $j \in J$ , the corresponding cost interval $[l_j,u_j]$ is
generated by uniformly selecting random integer values $l_{j}$ and $u_j$ in the ranges $[(1-\delta)\cdot \Phi_j,\Phi_j]$ and
 $[\Phi_j,(1+\delta)\cdot \Phi_j]$, respectively. These instances are named B.$<$\emph{SCinst}$>$-$\delta$, where $<$\emph{SCinst}$>$ stands for the name
of the original SC instance set considered. For each original instance of the classical SC, we considered three RSC instances, one for each value $\delta \in \{0.1, 0.3, 0.5\}$.
In total, 75 instances from Beasley benchmark were used in our experiments.

In Montemanni instances, the column costs of the original SC instances are discarded, and, for each column $j \in J$ , the corresponding cost interval
$[l_j,u_j]$ is generated as follows. First, a random integer value $u_{j}$ is uniformly
chosen in the range $[0,1000]$, and, then, a random integer value $l_j$ is uniformly selected in the range $[0, u_j]$.
These instances are named M.$<$\emph{SCinst}$>$-$1000$, where $<$\emph{SCinst}$>$ is the name
of the original SC instance set used. Each original SC instance considered gives the backbone to generate three RSC instances, and
a total of 75 instances from Montemanni benchmark were used in our experiments.

Kasperski-Zieli\'{n}ski benchmark also consider classical SC instances without the original column costs. In these instances,
the cost interval $[l_j,u_j]$ of each column $j \in J$  is generated as follows. First, a random integer value $l_{j}$ is uniformly
chosen in the range $[0,1000]$, and, then, a random integer value $u_j$ is uniformly selected in the range $[l_j, l_j + 1000]$.
These instances are named KZ.$<$\emph{SCinst}$>$-$1000$, where $<$\emph{SCinst}$>$ is the name
of the original SC instance set used. Each of the SC instances considered gives the backbone to generate three RSC instances, and
a total of 75 instances from Kasperski-Zieli\'{n}ski benchmark were used in our experiments.
In summary, 225 RSC instances were considered in the experiments.
\subsubsection{Results}
Table~\ref{table05} shows the computational results regarding the three benchmarks of RSC instances considered.
The first column displays the name of each instance set, while the second one gives (i) the number of instances solved at optimality by LB-Benders' within
3600 seconds, over (ii) the cardinality of the corresponding instance set. The average wall-clock processing time (in seconds) spent in solving these
instances at optimality are given in the third column. This entry is filled with a dash if no instance in the group was solved at optimality.
The fourth and fifth columns show, respectively, the average and the standard deviation (over all the instances in each set) of the relative optimality
gaps given by $100 \cdot \frac{UB_b - LB_b}{UB_b}$. Recall that $LB_b$ and $UB_b$ are, respectively, the best lower and upper bounds obtained by LB-Benders'
for a given instance. The sixth column displays the average wall-clock processing time (in seconds) of AMU. The seventh column shows the average
(over all the instances in each set) of the relative gaps given by $100 \cdot \frac{UB_{amu} - LB_b}{UB_{amu}}$, where $UB_{amu}$ is the best upper bound
obtained by AMU for a given instance. The standard deviation of these gaps is given in the eighth column.
Likewise, the ninth column shows the average wall-clock processing time of LPH, and the last two columns give the average and the standard deviation
(once again, over all the instances in each set) of the gaps given by $100 \cdot \frac{UB_{lph} - LB_b}{UB_{lph}}$.
Accordingly, $UB_{lph}$ is the robustness cost of the solution obtained by LPH for a given instance.

As for R-RSP, the gaps referred to the solutions obtained by AMU and LPH consider the (not necessarily optimal) lower bounds obtained by LB-Benders'
within 3600 seconds of execution. Thus, these gaps may overestimate the actual gaps between the cost of the solutions obtained and the cost of optimal ones.

From the results, it can be seen that the average gaps of the solutions provided by LPH are smaller than those referred to AMU for all the sets of instances
tested.
Furthermore, for the hardest instances (Kasperski-Zieli\'{n}ski benchmark), the average gaps of the solutions provided by LPH are even smaller than
those of LB-Benders'. As discussed in \cite{Pereira11}, Kasperski-Zieli\'{n}ski instances are especially challenging, mostly because of the high
probability of overlap between arbitrary cost intervals in these instances.
\begin{landscape}
\begin{table}[ht]
\caption{Computational results for the three benchmarks of RSC instances.}
\label{table05}
\center
\scalebox{0.7}{
\begin{tabular}{lrrrrrrrrrrrr}
\toprule
\multicolumn{1}{r}{} & \multicolumn{4}{c}{\textbf{LB-Benders'}} & & \multicolumn{3}{c}{\textbf{AMU}} & & \multicolumn{3}{c}{\textbf{LPH}}\\
\cmidrule{2-5} \cmidrule{7-9} \cmidrule{11-13}
Beasley set & {\#opt} & {Time (s)} & {AvgGAP (\%)} & {StDev (\%)} & & Time (s) & {AvgGAP (\%)} & {StDev (\%)} & & Time (s) & {AvgGAP (\%)} & {StDev (\%)}\\
\midrule
B.scp4-0.1 & 10/10 & 1.03~~ &0.00~~~~~~ & 0.00~~~~ & & 0.23~~ & 4.24~~~~~~  & 4.92~~~~ & &0.21~~ &0.59~~~~~~  &1.86~~~~~\\
B.scp5-0.1 & 10/10 & 3.32~~ &0.00~~~~~~ & 0.00~~~~ & & 0.34~~ & 9.01~~~~~~  & 8.27~~~~ & &0.42~~ &2.60~~~~~~  &3.62~~~~~\\
B.scp6-0.1 & 5/5 & 2.23~~ &0.00~~~~~~ & 0.00~~~~ & & 0.88~~ & 8.02~~~~~~  & 7.41~~~~ & &0.64~~ &1.43~~~~~~  &3.19~~~~~\\
B.scp4-0.3 & 10/10 & 9.82~~ &0.00~~~~~~ & 0.00~~~~ & & 0.24~~ & 5.24~~~~~~  & 3.46~~~~ & &0.44~~ &0.98~~~~~~  &1.12~~~~~\\
B.scp5-0.3 & 10/10 & 28.23~~ &0.00~~~~~~ & 0.00~~~~ & & 0.34~~ & 6.21~~~~~~  & 2.61~~~~ & &0.92~~ &1.78~~~~~~  &2.62~~~~~\\
B.scp6-0.3 & 5/5 & 12.18~~ &0.00~~~~~~ & 0.00~~~~ & & 1.08~~ & 2.44~~~~~~  & 2.70~~~~ & &1.05~~ &1.63~~~~~~  &1.51~~~~~\\
B.scp4-0.5 & 10/10 & 93.98~~ &0.00~~~~~~ & 0.00~~~~ & & 0.29~~ & 3.56~~~~~~  & 2.12~~~~ & &0.64~~ &1.20~~~~~~  &1.49~~~~~\\
B.scp5-0.5 & 10/10 & 94.20~~ &0.00~~~~~~ & 0.00~~~~ & & 0.34~~ & 5.38~~~~~~  & 4.33~~~~ & &1.21~~ &0.87~~~~~~  &1.00~~~~~\\
B.scp6-0.5 & 5/5 & 22.57~~ &0.00~~~~~~ & 0.00~~~~ & & 1.00~~ & 4.08~~~~~~  & 4.21~~~~ & &1.16~~ &1.44~~~~~~  &2.23~~~~~\\
\cmidrule{1-13}
\textbf{Average}       & & & 0.00~~~~~~ & 0.00~~~~ & & & 5.35~~~~~~ & 4.45~~~~~ & & &1.39~~~~~~ &2.07~~~~~ \\
\midrule
\\
\midrule
\multicolumn{1}{r}{} & \multicolumn{4}{c}{\textbf{LB-Benders'}} & & \multicolumn{3}{c}{\textbf{AMU}} & & \multicolumn{3}{c}{\textbf{LPH}}\\
\cmidrule{2-5} \cmidrule{7-9} \cmidrule{11-13}
Montemanni set & {\#opt} & {Time (s)} & {AvgGAP (\%)} & {StDev (\%)} & & Time (s) & {AvgGAP (\%)} & {StDev (\%)} & & Time (s) & {AvgGAP (\%)} & {StDev (\%)}\\
\midrule
M.scp4-1000 & 27/30 & 598.94~~ &0.04~~~~~~ & 0.14~~~~ & & 0.23~~ & 1.13~~~~~~  & 0.78~~~~ & &0.61~~ &0.05~~~~~~  &0.16~~~~~\\
M.scp5-1000 & 30/30 & 481.98~~ &0.00~~~~~~ & 0.00~~~~ & & 0.21~~ & 1.00~~~~~~  & 0.90~~~~ & &0.70~~ &0.01~~~~~~  &0.04~~~~~\\
M.scp6-1000 & 15/15 & 13.90~~ &0.00~~~~~~ & 0.00~~~~ & & 0.42~~ & 0.78~~~~~~  & 0.94~~~~ & &0.80~~ &0.06~~~~~~  &0.14~~~~~\\
\cmidrule{1-13}
\textbf{Average}       & & & 0.01~~~~~~ & 0.05~~~~ & & & 0.97~~~~~~ & 0.87~~~~~ & & &0.04~~~~~~ &0.11~~~~~\\
\midrule
\\
\midrule
\multicolumn{1}{r}{} & \multicolumn{4}{c}{\textbf{LB-Benders'}} & & \multicolumn{3}{c}{\textbf{AMU}} & & \multicolumn{3}{c}{\textbf{LPH}}\\
\cmidrule{2-5} \cmidrule{7-9} \cmidrule{11-13}
Kasperski-Zieli\'{n}ski set & {\#opt} & {Time (s)} & {AvgGAP (\%)} & {StDev (\%)} & & Time (s) & {AvgGAP (\%)} & {StDev (\%)} & & Time (s) & {AvgGAP (\%)} & {StDev (\%)}\\
\midrule
K.scp4-1000 & 0/30 & 0.00~~ &14.25~~~~~~ & 2.93~~~~ & & 3.08~~ & 14.27~~~~~~  & 2.73~~~~ & &124.86~~ &12.90~~~~~~  &2.62~~~~~\\
K.scp5-1000 & 0/30 & 0.00~~ &8.63~~~~~~ & 3.45~~~~ & & 2.90~~ & 8.63~~~~~~  & 3.45~~~~ & &44.93~~ &7.99~~~~~~  &3.01~~~~~\\
K.scp6-1000 & 3/15 & 1881.73~~ &2.95~~~~~~ & 3.13~~~~ & & 11.96~~ & 4.04~~~~~~  & 2.90~~~~ & &115.31~~ &2.77~~~~~~  &2.82~~~~~\\
\cmidrule{1-13}
\textbf{Average}       & & & 8.61~~~~~~ & 3.17~~~~ & & & 8.98~~~~~~ & 3.03~~~~~ & & &7.89~~~~~~ &2.82~~~~~ \\
\bottomrule
\end{tabular}
}
\end{table}
\end{landscape}

Notice that LPH clearly outperforms AMU in terms of the quality of the solutions obtained and requires significantly smaller computational effort than LB-Benders'.
In fact, for most of the instance sets considered, the average execution times of LPH are comparable to those of AMU.
Therefore, LPH arises as an effective and time efficient heuristic in solving the three benchmarks of RSC instances tested.
Nevertheless, we highlight that AMU should still be considered in the cases where time efficiency is preferable to the
quality of the solutions, especially for the more challenging instances.
\begin{table}[!h]
\caption{Details on the improvement of LPH over AMU for the three benchmarks of RSC instances.}
\label{table06}
\center
\scalebox{0.7}{
\begin{tabular}{lrrrrr}
\toprule
\multicolumn{2}{r}{} & \multicolumn{4}{c}{\textbf{LPH Improvement}}\\
\cmidrule{3-6}
Test set & {\#Better} & {Min (\%)} & {Max (\%)} & {Avg (\%)} & {StDev (\%)}\\
\midrule
B.scp4-0.1 & 10/10&0.00~~~~&12.50~~~~~~&3.66~~~~~&5.05~~~~~~~\\
B.scp5-0.1 & 10/10&0.00~~~~&20.83~~~~~~&6.53~~~~~&8.33~~~~~~~\\
B.scp6-0.1 & 5/5&0.00~~~~&15.00~~~~~~&6.69~~~~~&6.79~~~~~~~\\
B.scp4-0.3 & 10/10&0.00~~~~&11.76~~~~~~&4.29~~~~~&3.57~~~~~~~\\
B.scp5-0.3 & 10/10&0.00~~~~&8.06~~~~~~&4.47~~~~~&2.98~~~~~~~\\
B.scp6-0.3 & 4/5&-2.50~~~~&6.52~~~~~~&0.80~~~~~&3.37~~~~~~~\\
B.scp4-0.5 & 10/10&0.52~~~~&5.36~~~~~~&2.39~~~~~&1.42~~~~~~~\\
B.scp5-0.5 & 9/10&-1.52~~~~&13.64~~~~~~&4.52~~~~~&5.06~~~~~~~\\
B.scp6-0.5 & 5/5&0.00~~~~&6.58~~~~~~&2.70~~~~~&2.94~~~~~~~\\
\midrule
M.scp4-1000 & 30/30&0.17~~~~&3.30~~~~~~&1.08~~~~~&0.79~~~~~~~\\
M.scp5-1000 & 30/30&0.00~~~~&3.51~~~~~~&0.99~~~~~&0.88~~~~~~~\\
M.scp6-1000 & 15/15&0.00~~~~&2.87~~~~~~&0.72~~~~~&0.94~~~~~~~\\
\midrule
KZ.scp4-1000 & 29/30&-0.04~~~~&4.59~~~~~~&1.56~~~~~&1.33~~~~~~~\\
KZ.scp5-1000 & 29/30&-0.08~~~~&3.13~~~~~~&0.71~~~~~&0.80~~~~~~~\\
KZ.scp6-1000 & 14/15&-0.13~~~~&4.25~~~~~~&1.29~~~~~&1.61~~~~~~~\\
\bottomrule
\end{tabular}
}
\end{table}

In Table~\ref{table06}, we detail the improvement of LPH over AMU in terms of the bounds achieved for the three benchmarks of RSC instances.
The first column displays the name of each instance set, and the second one shows (i) the number of times LPH gives better bounds
than AMU, over (ii) the cardinality of the corresponding instance set. The third and fourth columns show, respectively, the minimum and the maximum percentage improvements
given by $100 \cdot \frac{UB_{amu} - UB_{lph}}{UB_{amu}}$, while the last two columns give the average and the standard deviation (over all the instances
in each set) of these values. These percentage improvements assume negative values whenever AMU gives better bounds than LPH.

Regarding Beasley instances, the percentage improvement is up to 20.83\% (see B.scp5-0.1). In fact, AMU only overcomes
LPH on two instances from the whole benchmark, one from B.scp6-0.3 and other from B.scp5-0.5.
With respect to Montemanni instances, the percentage improvement is up to 3.51\% (see M.scp5-1000), and LPH always gives better bounds
than AMU. Regarding Kasperski-Zieli\'{n}ski instances, the percentage improvement is at most 4.59\% (see KZ.scp4-1000), and
AMU overcomes LPH on three instances, one from each instance set.

We could not identify any specific pattern on the instances for which AMU outperforms LPH. However, such behaviour is only verified in
5 out of the 225 instances tested.

\section{Concluding remarks}
\label{s_remarks}

In this study, we proposed a novel heuristic approach to address the class of interval robust-hard problems, which is composed of
robust optimization versions of classical NP-hard combinatorial problems.
We applied this new technique to two interval robust-hard problems, namely the Restricted Robust Shortest Path problem (R-RSP)
and the Robust Set Covering problem (RSC). To our knowledge, the former problem was also introduced in this study. 

In order to evaluate the quality of the solutions obtained by our heuristic (namely LPH), we adapted to R-RSP and to RSC a logic-based Benders'
decomposition algorithm (referred to as LB-Benders'), which is a state-of-the-art exact method to solve interval data min-max regret problems in general.
We also compared the results obtained by the proposed heuristic with a widely used 2-approximation procedure, namely AMU. 

Regarding R-RSP, the computational experiments showed the heuristic's effectiveness in solving the two benchmarks of instances considered.
For the layered and acyclic digraph instances, the average relative gaps of LB-Benders' and AMU were, respectively, 4.49\% and 6.65\%, while that of LPH was
only 3.35\%. For the grid digraph instances, the average relative gaps of LB-Benders', AMU and LPH were, respectively, 1.15\%, 4.69\% and 1.26\%. Moreover,
the average processing times of LPH were much smaller than those of LB-Benders' for both benchmarks considered and remain competitive with those referred to AMU for most of the instances.
In addition, LPH was able to provide better bounds than AMU for 465 out of the 480 R-RSP instances tested.

With respect to RSC, the results show that the average gaps of the solutions provided by LPH were smaller than those referred to AMU for all the sets of
instances from the three benchmarks used in the experiments. Moreover, LPH achieved better primal bounds than LB-Benders' for the
hardest instances (Kasperski-Zieli\'{n}ski benchmark) and was able to provide better bounds than AMU for 220 out of the 225 RSC instances tested.

The results point out to the fact that the proposed heuristic framework may also be efficient in finding optimal or near-optimal solutions for
other interval robust-hard problems, which can be explored in future studies.
Other directions of future work remain open. For instance, research can be conducted to close the conjecture that the heuristic framework can be
applied to the wider class of interval data min-max regret problems with compact constraint sets addressed in \cite{Conde2010,Conde2012}.
Future works can also add to our framework local search strategies, such as Local Branching \cite{Fischetti03}, to further improve the quality of the
solutions obtained by the heuristic.
\section*{Acknowledgments}
This work was partially supported by the Brazilian National Council
for Scientific and Technological Development (CNPq), the Foundation
for Support of Research of the State of Minas Gerais, Brazil
(FAPEMIG), and Coordination for the Improvement of Higher Education
Personnel, Brazil (CAPES).
\section*{References}
 \bibliographystyle{elsarticle-num}
 \bibliography{main}


\end{document}